\documentclass[11pt]{article}
\usepackage[a4paper, margin=1in]{geometry}
\usepackage{amsmath, amssymb, amsfonts, amsthm}
\usepackage{graphicx}
\usepackage{hyperref}
\usepackage{multirow, multicol, comment, makecell}
\usepackage{tikz}
\usepackage{subcaption}
\usepackage{fancyhdr}
\usepackage{cleveref}

\newcommand{\qn}{\{0,1\}^n}

\theoremstyle{definition}
\newtheorem{theorem}{Theorem}[section]
\newtheorem{lemma}{Lemma}[section]
\newtheorem{proposition}{Proposition}[section]
\newtheorem{corollary}{Corollary}[section]
\newtheorem{definition}{Definition}[section]
\newtheorem{example}{Example}[section]
\newtheorem{remark}[theorem]{Remark}

\title{Hamiltonian dynamics of Boolean networks\thanks{This work was partially funded by ANID-Chile through the Center for Mathematical Modeling (CMM), proyecto BASAL FB210005, and ANID-Subdirección de Capital Humano/Magíster Nacional/2023-22231646.}}

\author{Arturo Zapata-Cortés\thanks{Departamento de Ingeniería Informática y Ciencias de la Computación, Universidad de Concepción, Chile. Email: \texttt{azapata2016@inf.udec.cl}} 
\and Julio Aracena\thanks{CI$^2$MA and Departamento de Ingeniería Matemática, Universidad de Concepción, Chile. Email: \texttt{jaracena@ing-mat.udec.cl}}}

\begin{document}

\maketitle

\begin{abstract}
    This article examines the impact of Hamiltonian dynamics on the interaction graph of Boolean networks. Three types of dynamics are considered: maximum height, Hamiltonian cycle, and an intermediate dynamic between these two. The study addresses how these dynamics influence the connectivity of the graph and the existence of variables that depend on all other variables in the system. Additionally, a family of unate Boolean networks capable of describing these three Hamiltonian behaviors is introduced, highlighting their specific properties and limitations. The results provide theoretical tools for modeling complex systems and contribute to the understanding of dynamic interactions in Boolean networks.
\end{abstract}

\section{Introduction} \label{chapter11}

Boolean networks (BNs) are a widely used mathematical model for representing complex systems composed of variables that can assume one of two possible states: \(0\) or \(1\). These networks have proven to be valuable tools in various fields such as biology \cite{thomas:hal-00087681}, genetics \cite{Huang2012, KAUFFMAN1969437, THOMAS1973563}, and social network theory \cite{ETESAMI2016100}, among others. By reducing problems to a binary context, BNs enable the modeling, simulation, and analysis of nonlinear interactions, as the study of the dynamics of systems characterized by multiple interdependent variables.

A significant portion of existing studies has focused on specific complex systems, emphasizing the interaction between variables to infer dynamic properties. Notable examples include the analysis of interaction graphs with bounded in-degree and their implications on dynamics \cite{citaa2}, the existence of fixed points \cite{aracena2008maximum, ARACENA2004277}, limit cycles \cite{PhysRevLett.94.088701, PhysRevLett.90.098701}, and the determination of the maximum length of limit cycles in certain families of BNs \cite{ARACENA2004237}.

However, most of these works rely on restrictions imposed on the interaction graph to infer dynamic properties, leaving the study of conditions induced in the interaction graph by a given dynamic largely unexplored.

The primary objective of this paper is to analyze the properties of the interaction graph induced by Hamiltonian dynamics \cite{mitesis}, characterized by a unique trajectory capable of visiting all states of the system. This analysis includes cases of maximum height, maximum limit cycle length, and dynamics intermediate to the two aforementioned cases.

Additionally, we address the problem: given a Hamiltonian digraph \(G_\Gamma\), is it possible to construct a unate BN whose dynamics is isomorphic to \(G_\Gamma\)? To understand this question, we explore certain families of BN capable of exhibiting Hamiltonian cycle behaviors and their implications for self-dual networks \cite{SATO1982201,MATAMALA1995251}. From this, we present a family of Hamiltonian unate Boolean networks, self-dual and non-neural.

This document is organized as follows: \Cref{ch2} introduces the fundamental definitions and notations. \Cref{chapter31} focuses on the analysis of maximum in-degree and connectivity in interaction graphs. In \Cref{chapter41}, a family of Hamiltonian unate BNs is introduced, and finally, \Cref{chapter51} presents the conclusions, discussing the obtained results and future work.

\section{Definitions and notation} \label{ch2}

A directed graph \(G = (V, A)\), where \(V\) is the set of vertices and \(A\) is the set of arcs. The in-degree of a vertex \(j \in V\) is denoted as \(d^-_G(j)\), and when there is no ambiguity, the subscript \(G\) is omitted. 
A path is an ordered sequence of vertices $P:v_0,v_1,\dots,v_k$ such that $(v_{i-1},v_i)\in A$ for all $i=1,\dots,k$.
A (strongly connected) {\bf component} is a maximal subset of vertices where, for any pair $u$ and $v$, there exists a path from $u$ to $v$ and vice versa. The \textbf{component graph} of $G$ is the directed acyclic graph whose vertices are the strongly connected components of $G$, and where there is an arc from $C_i$ to $C_j$ if and only if $G$ contains an arc from some $u \in C_i$ to some $v \in C_j$, with $i \neq j$.
A directed graph is said to have a \textbf{source component} if it is a component with no incoming arcs from other components of the graph, while a \textbf{sink component} is a component with no outgoing arcs to other components. Finally, two components are considered \textbf{independent} if there is no paths between them in \(G\). For a more detailed description of graph-related concepts, we recommend consulting \cite{Bang,west2001introduction}.

A \textbf{Boolean network} \(f: \{0,1\}^n \to \{0,1\}^n\), where \(n \in \mathbb{N}\), is a dynamic system defined in discrete time and space, consisting of \(n\) binary variables \(x_j\), with \(j \in [n]:= \{1,2,\ldots,n\}\). The network is described by Boolean functions $f_i: \{0,1\}^n \to \{0,1\}$, called \textbf{local activation functions}, where \(f = (f_1, f_2, \ldots, f_n)\) and \(x_j(t+1) = f_j(x(t))\) determines the temporal evolution of each variable.

The temporal evolution of the system is represented by a directed graph called the \textbf{state transition graph} or \textbf{dynamics} of \(f\), defined as follows:

\begin{equation*} \label{eq:000123}
    \Gamma(f) = (\{0,1\}^n, \{(x,f(x)) : x \in \{0,1\}^n\}).
\end{equation*}

Since \(f\) is a function, each vertex has an out-degree of one. We denote by \(\mathcal{G}(n)\) the family of digraphs isomorphic to the dynamics of some BN with \(n\) variables, described as follows:
\begin{equation*} \label{eq:000051}
    \mathcal{G}(n) = \{(V,A) \text{ digraph} : |V| = 2^n,\; \text{and for all } u \in V,\; d^+(u) = 1 \}.
\end{equation*}

The study of graphs \(G_\Gamma \in \mathcal{G}(n)\) aims to identify properties common to all BNs with dynamic behavior \(G_\Gamma\). The set of BNs whose dynamics is isomorphic to \(G_\Gamma\) is denoted by \(\mathcal{F}(G_\Gamma)\):
\begin{equation*} \label{eq:00010}
    \mathcal{F}(G_\Gamma) = \{f: f  \text{ is a Boolean network and }  \Gamma(f) \cong G_\Gamma \}.
\end{equation*}

We focus on digraphs \(G_\Gamma \in \mathcal{G}(n)\) that possess a path capable of visiting all their vertices, with the purpose of analyzing the properties of the family of BNs \(\mathcal{F}(G_\Gamma)\).

The configurations \(\vec{0}, \vec{1}, e_i \in \{0,1\}^n\) are defined as those with all zeros, all ones, and all zeros except for a one in component \(i \in [n]\), respectively. Additionally, \(\oplus\) denotes the modulo two sum operator, generalized to configurations in \(\{0,1\}^n\) by applying the operator component-wise. 

For a BN \(f\) with \(n\) variables and a configuration \(x \in \{0,1\}^n\), the following terms are defined:
A \textbf{Garden of Eden} is a configuration \(x\) such that \(f^{-1}(\{x\}) = \emptyset\). A configuration \(x\) is a \textbf{fixed point} if \(f(x) = x\), and it is \textbf{periodic} if there exists \(k \in \mathbb{N}\) such that \(f^k(x) = x\). Otherwise, \(x\) is called \textbf{transient}. A \textbf{limit cycle} is a cycle in \(\Gamma(f)\) of length at least two, and an \textbf{attractor} of the network is any fixed point or limit cycle.

Moreover, the \textbf{period} of \(f\), denoted as \(p(f)\), is the least common multiple of the lengths of all its limit cycles. The \textbf{height} of \(f\), denoted as \(h(f)\), is the smallest \(k \in \mathbb{N}\) such that, for any \(x \in \{0,1\}^n\), \(f^k(x)\) is a periodic point. Finally, a \textbf{trajectory} \(R\) of \(f\) is a path in \(\Gamma(f)\) that does not repeat arcs and \(|R|\) is the length of the trajectory, corresponding to the number of arcs.

We say that the local activation function \(f_j\) depends on the variable \(x_i\), or on the index \(i\), if there exists a configuration \(x \in \{0,1\}^n\) such that \(f_j(x) \neq f_j(x \oplus e_i)\). The \textbf{interaction graph} or \textbf{dependency graph} of the BN \(f\), denoted by \(G(f)\), is a directed graph with \(n\) vertices representing the network variables, where an edge \((i, j)\) indicates that \(f_j\) depends on \(x_i\).

Additionally, the \textbf{local interaction graph} at \(z\), denoted by \(G_z(f)\), is a subgraph of \(G(f)\) restricted to dependencies on a specific configuration \(z \in \{0,1\}^n\). The interaction and local interaction graphs are formally defined as:
\begin{align*}
    G(f) &= ([n], \{(i,j) \in [n] \times [n] : f_j \text{ depends on variable } x_i\}), \\
    G_z(f) &= ([n], \{(i,j) \in [n] \times [n] : f_j(z) \neq f_j(z \oplus e_i)\}).
\end{align*}

For any $I \subsetneq [n]$, the configuration $x_I \in \{0,1\}^{|I|}$ denotes the projection of $x$ onto the components indexed by $I$, while $\overline{x}^I \in \{0,1\}^n$ represents the negation of $x$ specifically on those components. In particular, we define the vectors with alternating values as $\overline{0}^E$ and $\overline{1}^E$, where $E = \{i \in \{1, \ldots, n\} : i \text{ is even}\}$.

A source component of $G(f)$ induces a dynamic that can be analyzed independently of the rest of the network. Let \(f\) be a BN of \(n\) variables, \(x \in \{0,1\}^n\), and \(I \subsetneq [n]\) a set that induces a source component in \(G(f)\). The \textbf{subnetwork of \(f\)} induced by \(I\) is the BN \(f_I : \{0,1\}^{|I|} \to \{0,1\}^{|I|}\), defined as \(f_I(x) = f(y)_I\), for any \(y \in \{0,1\}^{n}\) such that $y_I=x$.

\begin{example} \label{eje0003}
    Let \(f\) be the Boolean network of \(3\) variables described by the local activation functions and the network dynamics presented in \Cref{fig:inter}.

    In this case, the configurations \((1, 0, 0)\) and \((1, 0, 1)\) are Garden of Eden states, while \((0, 0, 0)\) and \((1, 1, 1)\) are fixed points. The limit cycle of the BN is \([(0, 1, 1),\) \((0, 0, 1)]\), with a period \(p(f) = 2\) and a height \(h(f) = 3\).
    \begin{align}
        f_3(1,1,1)&\neq f_3((1,1,1)\oplus e_3). \label{eq:00026}
    \end{align}
    Observe that \(f_1\) and \(f_2\) depend on all variables. On the other hand, (\ref{eq:00026}) shows that the local activation function \(f_3\) depends only on variable \(x_3\), since for any other index \(i \neq 3\), \(f_3(x) = f_3(x \oplus e_i)\) holds. The interaction graph \(G(f)\) along with $G_{(1,1,0)}(f)$ is shown in \Cref{fig:inter}.

    \begin{figure}[htb!]
    \begin{minipage}[htb!]{0.48\textwidth} 
        \centering
        \begin{align*}
            f_1(x)&=(x_1 \land \overline{x}_2\land \overline{x}_3) \lor (x_1\land x_2\land x_3)  \\
            f_2(x)&=(x_1 \land  x_3) \lor (\overline{x}_2 \land  x_3) \lor (x_1 \land  \overline{x}_2)  \\
            f_3(x)&= x_3
        \end{align*}
    \end{minipage}
    \hfill
    \begin{minipage}[htb!]{0.45\textwidth}
        \centering
        \begin{tikzpicture}[scale=0.85,every node/.style={circle, draw=black!80, fill=white, text=black, minimum size=6pt,inner sep=2pt,outer sep=2pt}]
            \node[fill opacity=0,draw opacity=0,text opacity=1] at (-110pt,40pt) {$\Gamma(f):$};
            \node (1) at (-80pt,20pt) {$100$};
            \node (2) at (-40pt,20pt) {$110$};
            \node (3) at (0pt,20pt) {$010$};
            \node (4) at (40pt,20pt) {$000$};
            \node (5) at (-80pt,-20pt) {$111$};
            \node (6) at (-40pt,-20pt) {$101$};
            \node (7) at (0pt,-20pt) {$011$};
            \node (8) at (40pt,-20pt) {$001$};
            \path[very thick, -to]
            (1) edge[black] (2)
            (2) edge[black] (3)
            (3) edge[black] (4)
            (6) edge[black] (7)
            (7) edge[bend right=15,black] (8)
            (8) edge[bend right=15,black] (7)
            (5.-170.00) edge[controls=+(-140.00:30pt)and +(140.00:30pt),black] (5.170.00)
            (4.-10.00) edge[controls=+(-40.00:30pt)and +(40.00:30pt),black] (4.10.00)
            ;
        \end{tikzpicture}
        \end{minipage}
        \centering
        \begin{tikzpicture}[scale=0.85,every node/.style={circle, draw=black!60, fill=white, text=black, minimum size=6pt,inner sep=4pt,outer sep=2pt}]
            \node[fill opacity=0,draw opacity=0,text opacity=1] at (-150pt,30pt) {$G(f):$};
            \node (0) at (-90pt,20pt) {$1$};
            \node (1) at (-125pt,-35pt) {$2$};
            \node (2) at (-60pt,-35pt) {$3$};
            \node[fill opacity=0,draw opacity=0,text opacity=1] at (30pt,30pt) {$G_{(1,1,0)}(f):$};
            \node (3) at (90pt,20pt) {$1$};
            \node (4) at (60pt,-35pt) {$2$};
            \node (5) at (125pt,-35pt) {$3$};
            \path[ultra thick, -to]
            (2.-50.00) edge[controls=+(-70.00:30pt)and +(20.00:30pt),black] (2.-10.00)
            (0) edge[bend right=15,black] (1)
            (1) edge[bend right=15,black] (0)
            (0.70.00) edge[controls=+(40.00:30pt)and +(130.00:30pt),black] (0.110.00)
            (2) edge[bend right=15,black] (0)
            (2) edge[bend right=15,black] (1)
            (1.-160.00) edge[controls=+(-190.00:30pt)and +(-100.00:30pt),black] (1.-120.00)
            (5.-50.00) edge[controls=+(-70.00:30pt)and +(20.00:30pt),black] (5.-10.00)
            (4) edge[bend right=15,black] (3)
            (3) edge[bend right=15,black] (4)
            (5) edge[bend right=15,black] (3)
            ;
        \end{tikzpicture}
        \caption{Boolean network \(f\) from \Cref{eje0003}.}
        \label{fig:inter}
    \end{figure}
\end{example}

\begin{definition} \label{def:00021}
    For \(x, y \in \{0,1\}^n\), we define \(x \leq y\) if, for every component, \(x_i \leq y_i\) holds. 
    Given \(f\), a Boolean network with \(n \in \mathbb{N}\) variables, the local activation function \(f_j\) is said to be:
    \begin{itemize}
        \item \textbf{Increasing in component} \(i \in [n]\), if for any configuration \(x\) such that \(x_i = 0\), it holds that \(f_j(x) \leq f_j(x \oplus e_i)\).
        \item \textbf{Decreasing in component} \(i \in [n]\), if for any configuration \(x\) such that \(x_i = 0\), it holds that \(f_j(x) \geq f_j(x \oplus e_i)\).
    \end{itemize}
    Additionally, \(f_j\) is called \textbf{unate} if it is either increasing or decreasing in each of its components \(i \in [n]\). 
    The network is said to be \textbf{unate} if all its local activation functions are unate.
    Furthermore, the network is said to be \textbf{monotone} if all its local activation functions are increasing in each of its components \(i \in [n]\). The local activation function \(f_j\) is called {\bf threshold} (or linearly separable) if and only if exist real numbers $a_1,a_2,...,a_n,b$ such that $f=1$ if the sum of all $a_ix_i$, $1\leq i\leq n$, is greater than or equal to $b$, and $f=0$ otherwise. A {\bf neural network} is a Boolean network where all of its local activation functions are thresholds.
\end{definition}

The arcs of the interaction graph of a unate BN $f$ can be labeled with \textbf{signs} \(\sigma_f(i, j) \in \{+1, -1\}\), which indicate the nature of the relationship between the variables. A positive sign (\(+1\)) implies that \(f_j\) is increasing with respect to \(x_i\), while a negative sign (\(-1\)) indicates that \(f_j\) is decreasing with respect to \(x_i\). The above definition generalizes to the local interaction graph $G_z(f)$.

\begin{definition} \label{defi0010}
    A Boolean network \(f\) is said to be \textbf{Hamiltonian} if its dynamics possess a trajectory that reaches all configurations, and hence it has a unique attractor. A Hamiltonian Boolean network \(f\) is classified as \textbf{maximum height} if its only attractor is a fixed point; \textbf{intermediate height} if its unique attractor is a limit cycle of length \(k \in \{2, 3, \ldots, 2^n - 1\}\); or a \textbf{Hamiltonian cycle} if its dynamics form a limit cycle of length \(2^n\).
    Similarly, a digraph \(G_\Gamma \in \mathcal{G}(n)\) is classified as a Hamiltonian of maximum height, intermediate height, or Hamiltonian cycle if \(G_\Gamma \cong \Gamma(f)\) and \(f\) belong to the corresponding classification.
\end{definition}

\begin{example} \label{eje0004}
    Let \(f = (f_1, f_2, f_3)\) be a Boolean network with local activation functions and the interaction graph described in \Cref{fig:altej}. The Boolean network is Hamiltonian of maximum height, and its dynamics is shown in \Cref{fig:altej}.
    
    \begin{figure}[htb!]
    \begin{minipage}[htb!]{0.47\textwidth}
        \centering
        \begin{align*}
            f_1(x)&=x_3  \\
            f_2(x)&=\overline{x}_1 \lor(x_1 \land x_2 \land \overline{x}_3)  \\
            f_3(x)&=(x_2 \land x_3)\lor( \overline{x}_2 \land \overline{x}_3)
        \end{align*}
    \end{minipage}
    \hfill
    \begin{minipage}[htb!]{0.45\textwidth}
        \centering
        \begin{tikzpicture}[scale=0.85,every node/.style={circle, draw=black!60, fill=white, text=black, minimum size=6pt,inner sep=4pt,outer sep=2pt}]
            \node[draw=white, fill=white,text=black](n) at (-80pt,35pt) {$G(f):$};
            \node (0) at (35pt,-35pt) {$3$};
            \node (1) at (0pt,25pt) {$1$};
            \node (2) at (-35pt,-35pt) {$2$};
            \path[ultra thick, -to]
            (0.-50.00) edge[controls=+(-70.00:30pt)and +(20.00:30pt),black] (0.-10.00)
            (0) edge[bend right=15,black] (1)
            (1) edge[bend right=15,black] (2)
            (2.-160.00) edge[controls=+(-190.00:30pt)and +(-100.00:30pt),black] (2.-120.00)
            (0) edge[bend right=15,black] (2)
            (2) edge[bend right=15,black] (0)
            ;
        \end{tikzpicture}
        \end{minipage}
        \centering
        \begin{tikzpicture}[scale=1,every node/.style={circle, draw=black!60, fill=white, text=black, minimum size=6pt,inner sep=2.5pt,outer sep=2pt}]
            \node[draw=white, fill=white,text=black](n) at (-200pt,40pt) {$\Gamma(f):$};
            \node (1) at (-160pt,0pt) {$000$};
            \node (2) at (-120pt,0pt) {$011$};
            \node (3) at (-80pt,0pt) {$111$};
            \node (4) at (-40pt,0pt) {$101$};
            \node (5) at (0pt,0pt) {$100$};
            \node (6) at (40pt,0pt) {$001$};
            \node (7) at (80pt,0pt) {$110$};
            \node (8) at (120pt,0pt) {$010$};
            \path[very thick, -to]
            (1) edge[black] (2)
            (2) edge[black] (3)
            (3) edge[black] (4)
            (4) edge[black] (5)
            (5) edge[black] (6)
            (6) edge[black] (7)
            (7) edge[black] (8)
            (8.-10.00) edge[controls=+(-30.00:30pt)and +(40.00:30pt),black] (8.10.00)
            ;
        \end{tikzpicture}
        \caption{Boolean network \(f\) from \Cref{eje0004}.}
        \label{fig:altej}
    \end{figure}
\end{example}

\begin{example} \label{eje0005}
    Given the Boolean network \(f = (f_1, f_2, f_3)\) defined in \Cref{fig:hibridej}, along with its interaction graph, we observe that the network is Hamiltonian of intermediate height, as reflected in its dynamics described in \Cref{fig:hibridej}.
    
    \begin{figure}[htb!]
    \begin{minipage}[htb!]{0.47\textwidth} 
        \centering
        \begin{align*}
            f_1(x)=&\overline{x}_1 , \\
            f_2(x)=&(\overline{x}_2 \land x_3) \lor (\overline{x}_1 \land x_3)\lor
            (x_1\land x_2) ,\\
            f_3(x)=& \overline{x}_2 .
        \end{align*}
    \end{minipage}
    \hfill
    \begin{minipage}[htb!]{0.45\textwidth}
        \centering
        \begin{tikzpicture}[scale=0.85,every node/.style={circle, draw=black!60, fill=white, text=black, minimum size=6pt,inner sep=4pt,outer sep=2pt}]
            \node[draw=white, fill=white,text=black](n) at (-75pt,50pt) {$G(f):$};
            \node (0) at (35pt,-35pt) {$3$};
            \node (1) at (0pt,25pt) {$1$};
            \node (2) at (-35pt,-35pt) {$2$};
            \path[ultra thick, -to]
            (1.70.00) edge[controls=+(40.00:30pt)and +(130.00:30pt),black] (1.110.00)
            (2.-160.00) edge[controls=+(-190.00:30pt)and +(-100.00:30pt),black] (2.-120.00)
            (1) edge[bend right=15,black] (2)
            (2) edge[bend right=15,black] (0)
            (0) edge[bend right=15,black] (2)
            ;
        \end{tikzpicture}
    \end{minipage}
        \centering
        \begin{tikzpicture}[scale=0.9,every node/.style={circle, draw=black!60, fill=white, text=black, minimum size=6pt,inner sep=2.5pt,outer sep=2pt}]
            \node[draw=white, fill=white,text=black](n) at (-170pt,50pt) {$\Gamma(f):$};
            \node (1) at (-140pt,0pt) {$000$};
            \node (2) at (-100pt,0pt) {$101$};
            \node (3) at (-60pt,0pt) {$011$};
            \node (4) at (-20pt,0pt) {$110$};
            \node (5) at (20pt,0pt) {$010$};
            \node (6) at (60pt,30pt) {$111$};
            \node (7) at (100pt,0pt) {$001$};
            \node (8) at (60pt,-30pt) {$100$};
            \path[very thick, -to]
            (1) edge[black] (2)
            (2) edge[black] (3)
            (3) edge[black] (4)
            (4) edge[black] (5)
            (5) edge[bend right=30,black] (8)
            (8) edge[bend right=30,black] (7)
            (7) edge[bend right=30,black] (6)
            (6) edge[bend right=30,black] (5)
            ;
        \end{tikzpicture}
        \caption{Boolean network \(f\) from \Cref{eje0005}.}
        \label{fig:hibridej}
    \end{figure}
\end{example}

\begin{example} \label{eje0006}
    Given \(f = (f_1, f_2, f_3)\), described by the local activation functions and the interaction graph shown in \Cref{figciclej}, it follows that \(f\) is Hamiltonian with a cycle, as its dynamics form a cycle of length \(2^3\), as depicted in \Cref{figciclej}.
    
    \begin{figure}[htb!]
    \begin{minipage}[htb!]{0.4\textwidth} 
        \centering
        \begin{align*}
            f_1(x)=&\overline{x}_1,  \\
            f_2(x)=&(\overline{x}_2 \land x_3) \lor (\overline{x}_1 \land x_3)\lor
            (x_1\land x_2\land \overline{x}_3), \\
            f_3(x)=&\overline{x}_2 .
        \end{align*}
    \end{minipage}
    \hfill
    \begin{minipage}[h]{0.4\textwidth}
        \centering
        \begin{tikzpicture}[scale=0.85,every node/.style={circle, draw=black!60, fill=white, text=black, minimum size=6pt,inner sep=4pt,outer sep=2pt}]
            \node[draw=white, fill=white,text=black](n) at (-65pt,40pt) {$G(f):$};
            \node (0) at (35pt,-35pt) {$3$};
            \node (1) at (0pt,25pt) {$1$};
            \node (2) at (-35pt,-35pt) {$2$};
            \path[ultra thick, -to]
            (1.70.00) edge[controls=+(40.00:30pt)and +(130.00:30pt),black] (1.110.00)
            (2.-160.00) edge[controls=+(-190.00:30pt)and +(-100.00:30pt),black] (2.-120.00)
            (1) edge[bend right=15,black] (2)
            (2) edge[bend right=15,black] (0)
            (0) edge[bend right=15,black] (2)
            ;
        \end{tikzpicture}
    \end{minipage}
        \centering
        \begin{tikzpicture}[scale=0.85,every node/.style={circle, draw=black!60, fill=white, text=black, minimum size=6pt,inner sep=2.5pt,outer sep=2pt}]
            \node[draw=white, fill=white,text=black](n) at (-90pt,60pt) {$\Gamma(f):$};
            \node (1) at (0pt,60pt) {$000$};
            \node (2) at (-45pt,45pt) {$101$};
            \node (3) at (-65pt,0pt) {$011$};
            \node (4) at (-45pt,-45pt) {$110$};
            \node (5) at (0pt,-60pt) {$010$};
            \node (6) at (45pt,-45pt) {$100$};
            \node (7) at (65pt,0pt) {$001$};
            \node (8) at (45pt,45pt) {$111$};
            \path[very thick, -to]
            (1) edge[bend right=15,black] (2)
            (2) edge[bend right=15,black] (3)
            (3) edge[bend right=15,black] (4)
            (4) edge[bend right=15,black] (5)
            (5) edge[bend right=15,black] (6)
            (6) edge[bend right=15,black] (7)
            (7) edge[bend right=15,black] (8)
            (8) edge[bend right=15,black] (1)
            ;
        \end{tikzpicture}
        \caption{Boolean network \(f\) from \Cref{eje0006}.}
        \label{figciclej}
    \end{figure}
\end{example}

Our objective is on the properties of a BN with Hamiltonian dynamic. Next, we define a partition of the configuration space to identify patterns in the variable behavior that will be useful in the following results.

\begin{definition}[\cite{crama2011boolean}] \label{defin0002}
    Given \(f\), a Boolean network with \(n\) variables, and \(j \in [n]\), the set \(T(f_j) \subseteq \{0,1\}^n\) is defined as the set of true points, and \(F(f_j) \subseteq \{0,1\}^n\) as the set of false points of \(f_j\). These sets are defined as follows:
    \begin{align}
        T(f_j) &= \{x \in \{0,1\}^n : f_j(x) = 1 \}, \label{ccp} \\
        F(f_j) &= \{x \in \{0,1\}^n : f_j(x) = 0 \}. \label{ccn}
    \end{align}
    If, for every \(j \in [n]\), it holds that \(|T(f_j)| = |F(f_j)| = 2^{n-1}\), the Boolean network \(f\) is said to be \textbf{balanced}.
\end{definition}

The sets of true and false points establish a partition of \(\{0,1\}^n\), implying that for any \(j \in [n]\), \(|T(f_j)| + |F(f_j)| = 2^n\) holds. Henceforth, results concerning the set \(T\) will also apply to the set of false points \(F\).

\section{In-degree and connectivity of the interaction graph} \label{chapter31}

In \cite{WOS:001028245500001}, it is established that any digraph $G_\Gamma \in \mathcal{G}(n)$, other than the constant and identity digraphs, admits a BN $f \in \mathcal{F}(G_\Gamma)$ whose interaction graph is $K_n$, a complete digraph including loops. Furthermore, \cite{bridoux20244} demonstrates the existence of dynamics that impose constraints on the interaction graph, such as those involving a certain edge density or requiring a minimum in-degree of 2, among others. To further explore this direction, \Cref{lemaimpar} establishes a necessary condition regarding the in-degree of the interaction graph, thereby linking BNs to their dynamical behavior.

\begin{lemma}[\cite{10.1093/pnasnexus/pgac017,mitesis}] \label{lemaimpar}
    Let \(f\) be a Boolean network with \(n \in \mathbb{N}\) variables and \(j \in [n]\) such that \(|T(f_j)|\) is odd. Then, the in-degree of vertex \(j\) in the interaction graph is \(n\).
\end{lemma}
\begin{proof}
    Suppose \(|T(f_j)|\) is odd and that \(f_j\) does not depend on the index \(i \in [n]\). From the definition of dependence, it follows that for any configuration \(x \in T(f_j)\), \mbox{\(x \oplus e_i \in T(f_j)\)}. This contradicts the hypothesis that \(|T(f_j)|\) is odd, completing the proof.
\end{proof}

Importantly, the converse implication of \Cref{lemaimpar} is not true: in \Cref{eje0006}, \(T(f_2)\) is a set of even cardinality, although variable \(2\) has maximum in-degree. 
Conversely, if $G_\Gamma \in \mathcal{G}(n)$ possesses a unique vertex with in-degree zero, the associated dynamics deviate from bijective behavior by exactly one image. This observation motivates \Cref{prop:prop507}, as previously established in \cite{mitesis,citaa2}.

\begin{proposition} \label{prop:prop507}
    If \(G_\Gamma \in \mathcal{G}(n)\) (not necessarily connected) has exactly one vertex with in-degree zero, then 
    for any Boolean network \(f \in \mathcal{F}(G_\Gamma)\), there exists a component \(j \in [n]\) such that \(d^{-}(j) = n\) in its interaction graph.
\end{proposition}
\begin{proof}
    Suppose \(G_\Gamma\) has exactly one vertex with in-degree zero. By definition, the out-degree of every vertex in \(G_\Gamma\) is one. 
    Consequently, \(2^n-2\) vertices have in-degree one, one vertex has in-degree zero, and one vertex has in-degree two.
    
    For an arbitrary Boolean network \(f \in \mathcal{F}(G_\Gamma)\), denote \(u, v \in \{0,1\}^n\) as the Garden of Eden and the configuration 
    with two preimages, respectively. Since these are distinct configurations, there exists a component \(j \in [n]\) such that \(u_j \neq v_j\). 
    Let us analyze the cases based on the value of component \(j\) of \(u\).

    \begin{itemize}
        \item If \(u_j = 1\), given that the dynamics have \(2^n-2\) configurations with exactly one preimage and \(v_j = \overline{u_j} = 0\), 
        it follows that \(|T(f_j)| = 2^{n-1} - 1\), which is odd. 
        \item If \(u_j = 0\), it follows that \(|T(f_j)| = 2^{n-1} + 1\), which is also odd.
    \end{itemize}
    In both cases, \(T(f_j)\) has odd cardinality, and by \Cref{lemaimpar}, it is concluded that for any Boolean network \(f \in \mathcal{F}(G_\Gamma)\), there exists a vertex \(j \in [n]\) with in-degree \(d^{-}(j) = n\) in \(G(f)\).
\end{proof}

\Cref{prop:prop507} provides a sufficient condition to guarantee the existence of a variable with in-degree \(n\) in \(G(f)\) which is satisfied by Hamiltonian BNs of maximum and intermediate height. However, this result does not apply to Hamiltonian cycle, as there exists a counterexample\footnote{Courtesy of \textbf{Florian Bridoux} (\textbf{Univ. Côte d’Azur, CNRS, I3S UMR 7271, Sophia Antipolis, France}), personal communication, \textbf{2024}.} where no variable reaches this degree for $n= 5$ and has the following local activation functions.
\begin{align*}
    f_1(x)=&(\lnot x_2\land \lnot x_3)
    \lor(\lnot x_4\land x_5)
    \lor(\lnot x_3\land x_4\land \lnot x_5), \\
    f_2(x)=&(\lnot x_1\land \lnot x_3)
    \lor(x_1\land x_3\land x_4)
    \lor(\lnot x_1\land \lnot x_4\land x_5)
    \lor(\lnot x_1\land x_4\land \lnot x_5)\lor\\ 
    &(x_3\land x_4\land \lnot x_5), \\
    f_3(x)=&(\lnot x_1\land \lnot x_2\land x_4)
    \lor(x_1\land x_2\land x_5)
    \lor(\lnot x_1\land \lnot x_2\land x_5)
    \lor(x_1\land \lnot x_2\land \lnot x_5)
    \\ 
    &\lor(\lnot x_2\land x_4\land \lnot x_5)
    \lor(\lnot x_1\land x_2\land \lnot x_4\land \lnot x_5), \\
    f_4(x)=&(\lnot x_1\land x_2\land x_3)
    \lor(x_1\land \lnot x_2\land x_3)
    \lor(\lnot x_1\land x_2\land \lnot x_5)
    \lor(x_1\land \lnot x_2\land \lnot x_5)
    \\ 
    &\lor(x_1\land \lnot x_3\land \lnot x_5)
    \lor(x_2\land \lnot x_3\land \lnot x_5) 
    \lor(\lnot x_1\land \lnot x_2\land \lnot x_3\land x_5), \\
    f_5(x)=&(x_2\land x_3)
    \lor(x_1\land x_2\land x_4)
    \lor(\lnot x_1\land x_2\land \lnot x_4)
    \lor(x_1\land \lnot x_3\land x_4)
    \lor\\ 
    &(\lnot x_2\land \lnot x_3\land x_4). 
\end{align*}

The proof of \Cref{prop:prop507} allows the identification of variables with in-degree \(n\), using the labels assigned to both the Garden of Eden and the configuration with an odd number of preimages.

\begin{remark} \label{rmk:garden}
    By \Cref{prop:prop507} any Boolean network with a unique Garden of Eden has a connected interaction graph, as is the case of maximum and intermediate height Hamiltonian dynamics. 
\end{remark}

In Hamiltonian Boolean networks, there is a relationship between global dynamics and the induced subnetworks. 
In particular, this relationship extends to any \(G_\Gamma \in \mathcal{G}(n)\) that has a sufficiently long trajectory.

\begin{lemma} \label{prop799}
    Let \(f\) be a Boolean network with \(n \geq 2\) variables, and let \(I \subsetneq [n]\) be a non-empty set that induces a subnetwork of \(f\), that is, $I$ induces a source component of $G(f)$. 
    If the network has a trajectory in its dynamics of length greater than \(2^n - 2^{n - |I|}+1\), then \(f_I\) is Hamiltonian cycle.
\end{lemma}
\begin{proof}
 
    Let $P$ be a trajectory in $\Gamma(f)$ of length exceeding $2^n - 2^{n - |I|}+1$, and let $z \in \{0,1\}^n$ denote its initial configuration. By contradiction, assume there exists no $k \in \mathbb{N}$ such that $(f_I)^k(z_I) = z_I$. This implies that the trajectory $P$ contains no configuration $w \neq z \in \{0,1\}^n$ satisfying $w_I = z_I$. However, this would mean that $P$ does not contain $2^{n-|I|} - 1$ configurations, leading to $|P| \leq 2^n - 2^{n-|I|} + 1$, which contradicts our initial assumption.

    Since the cycle length in $f_I$ is at most $2^{|I|}$, let $k \in \mathbb{N}$ be the minimal value such that $(f_I)^k(z_I) = z_I$, where $k \leq 2^{|I|}$. We shall demonstrate that $k = 2^{|I|}$. Given that $P$ contains no repeated internal vertices and the projection $z_I$ occurs at most $2^{n - |I|}$ times along the trajectory, the length of $P$ is bounded as follows:
    \begin{equation} \label{eq:001231}
        2^n - 2^{n - |I|}+1 < |P| < k \cdot 2^{n - |I|}+1.
    \end{equation}
    Noting that $2^n - 2^{n - |I|} = (2^{|I|} - 1) \cdot 2^{n - |I|}$, it follows from (\ref{eq:001231}) that $k > 2^{|I|} - 1$. Since $k$ is an integer, this necessarily implies $k = 2^{|I|}$, which shows that $f_I$ induces a Hamiltonian cycle.
\end{proof}

Dynamics that describe a long limit cycle are of interest, for example, in cryptography (\cite{citaa2,6897028}) where the generation of pseudorandom sequences with low computational cost is essential. In this context, linear local activation functions and the {\bf quasi-Hamiltonian} Boolean networks are described, whose dynamics consist of a cycle of length \(2^n-1\) and a fixed point. Although the dynamics of this network is neither strictly Hamiltonian nor connected, it satisfies \Cref{prop799} and is of theoretical interest for being modelable with a bounded in-degree interaction graph \cite{citaa2}. \Cref{fig:ej0002} describe an example of these dynamics.

\begin{figure}[htb!]
    \begin{minipage}[htb!]{0.53\textwidth} 
        \centering
        \begin{align*}
            f_1(x)=&(x_2\land \overline{x}_3) \lor (\overline{x}_2\land x_3), \\
            f_2(x)=&x_3, \\
            f_3(x)=&x_1.
        \end{align*}
    \end{minipage}
    \hfill
    \begin{minipage}[htb!]{0.45\textwidth}
        \centering
        \begin{tikzpicture}[scale=0.9,every node/.style={circle, draw=black!60, fill=white, text=black, minimum size=6pt,inner sep=4pt,outer sep=2pt}]
            \node[draw=white, fill=white,text=black](n) at (-65pt,40pt) {$G(f):$};
            \node (0) at (35pt,-35pt) {$3$};
            \node (1) at (0pt,25pt) {$1$};
            \node (2) at (-35pt,-35pt) {$2$};
            \path[ultra thick, -to]
            (1) edge[bend left=15,black] (0)
            (0) edge[bend left=15,black] (2)
            (2) edge[bend left=15,black] (1)
            (0) edge[bend left=15,black] (1)
            ;
        \end{tikzpicture}
    \end{minipage}
        \centering
        \begin{tikzpicture}[scale=0.85,every node/.style={circle, draw=black!60, fill=white, text=black, minimum size=6pt,inner sep=2.5pt,outer sep=2pt}]
            \node[draw=white, fill=white,text=black](n) at (-100pt,60pt) {$\Gamma(f):$};
            \node (1) at (130pt,0pt) {$000$};
            \node (2) at (-30pt,45pt) {$101$};
            \node (3) at (-65pt,0pt) {$111$};
            \node (4) at (-45pt,-45pt) {$011$};
            \node (5) at (0pt,-60pt) {$010$};
            \node (6) at (45pt,-45pt) {$100$};
            \node (7) at (65pt,0pt) {$001$};
            \node (8) at (30pt,45pt) {$110$};
            \path[very thick, -to]
            (1.-10.00) edge[controls=+(-30.00:30pt)and +(40.00:30pt),black] (1.10.00)
            (2) edge[bend right=15,black] (3)
            (3) edge[bend right=15,black] (4)
            (4) edge[bend right=15,black] (5)
            (5) edge[bend right=15,black] (6)
            (6) edge[bend right=15,black] (7)
            (7) edge[bend right=15,black] (8)
            (8) edge[bend right=15,black] (2)
            ;
        \end{tikzpicture}
        \caption{Quasi-Hamiltonian dynamics.}
        \label{fig:ej0002}
\end{figure}

\begin{remark}
    By \Cref{prop799}, for any Hamiltonian or quasi-Hamiltonian BN \(f\) with \(n \geq 2\) and proper subset \(I \subsetneq [n]\) inducing a subnetwork, \(f_I\) has a Hamiltonian cycle.
\end{remark}

\Cref{prop799} strengthens the connection between the dynamics digraph and the resulting interaction graph, highlighting how the global properties of a network influence the characteristics of its induced subnetworks. Next, we explore how these properties affect the connectivity of the interaction graph.

\begin{proposition} \label{conexidad}
    If \(G_{\Gamma}\) is Hamiltonian, then \(G(f)\) is connected for any \(f \in \mathcal{F}(G_{\Gamma})\).
\end{proposition}
\begin{proof}
    If \(G_{\Gamma}\) is Hamiltonian of maximum or intermediate height, the result of \Cref{prop:prop507} is obtained.
    
    Suppose \(G_{\Gamma}\) is a Hamiltonian cycle, and consider \(f \in \mathcal{F}(G_{\Gamma})\). 
    By contradiction, assume \(G(f)\) is not connected. This implies that \(G(f)\) has \(k \geq 2\) connected components, denoted as \(G(f)[S_1], G(f)[S_2], \ldots, G(f)[S_k]\), with \(S_i \subsetneq [n]\) for \(i \in \{1, 2, \ldots, k\}\). Since \(G(f)[S_i]\) has no incoming or outgoing edges to other components, by \Cref{prop799}, each subnetwork induced by \(S_i\) is a Hamiltonian cycle.
    
    Let \(d \in \mathbb{N}\) denote the least common multiple described in (\ref{eq00020}). Since \(f_{S_i}\) is a Hamiltonian cycle, its period \(p(f_{S_i})\) is a power of two, and the least common multiple of these periods corresponds to the largest of these powers.
    \begin{align} \label{eq00020}
        d &= \operatorname{lcm} \{p(f_{S_i}) : i \in \{1, 2, \ldots, k\} \} \\
        &= \operatorname{max} \{2^{|S_i|} : i \in \{1, 2, \ldots, k\} \}.
    \end{align}
    Since \(|S_i| < n\), it follows that \(d < 2^n\). For an arbitrary configuration \(x \in \{0,1\}^n\), observe that 
    \((f_{S_i})^d(x_{S_i}) = x_{S_i}\) for all \(i \in \{1, 2, \ldots, k\}\). This implies that the period \(p(f) \leq d < 2^n\), contradicting the assumption that \(f\) is a Hamiltonian cycle. Thus, for any \(f \in \mathcal{F}(G_{\Gamma})\), the digraph \(G(f)\) must be connected.
\end{proof}

\Cref{eje0005} and \ref{eje0006} present Hamiltonian cycle and intermediate height BNs with interaction graphs that are not strongly connected. 
In the case that \(G(f)\) is not strongly connected, the graph is shown to be unilaterally connected.

\begin{definition}
    A digraph is \textbf{unilaterally connected} if, for every pair of vertices $\{x, y\}$, there exists at least one path from $x$ to $y$ or from $y$ to $x$.
\end{definition}

It is straightforward to observe that a digraph is unilaterally connected if and only if it has no independent components. This fact is fundamental to the proof of \Cref{prop717}, which establishes the connection between Hamiltonian dynamics and the structure of the interaction graph.

\begin{theorem} \label{prop717}
    Let \(f\) be a Boolean network with Hamiltonian or quasi-Hamiltonian dynamics. The interaction graph of \(f\) is unilaterally connected.
\end{theorem}
\begin{proof}
    Assume, by contradiction, that \(f\) has an interaction graph that is not unilaterally connected. Then, there exist two vertices in \(G(f)\) with no paths between them, implying they belong to independent components induced by \(A, B \subsetneq [n]\).

    Let \(M \subsetneq [n]\) be a set inducing a subnetwork of \(f\) that has edges directed to vertices in \(A\) and \(B\) in \(G(f)\) but not the other way around. This implies that \(M \cup A\), \(M \cup B\), and \(M \cup A \cup B =: D\) induce Hamiltonian cycle subnetworks due to \Cref{prop799}.
    
    Denote \(g = f_D\) and let \(s = \max\{p(f_{M \cup A}), p(f_{M \cup B})\} < 2^{|D|}\). Since \(p(f_{M \cup A})\) and \(p(f_{M \cup B})\) are multiples of a power of two, for any configuration \(x \in \{0,1\}^{|D|}\), it holds:
    \begin{align} \label{eq0020}
        g^s(x) &= (g^s(x)_A, (g_{M \cup B})^s(x_{M \cup B})) = (g^s(x)_A, x_{M \cup B}) = (g^s(x)_A, x_M, x_B), \\ 
        g^s(x) &= ((g_{M \cup A})^s(x_{M \cup A}), g^s(x)_B) = (x_{M \cup A}, g^s(x)_B) = (x_A, x_M, g^s(x)_B). \label{eq0021}
    \end{align}
    From (\ref{eq0020}) and (\ref{eq0021}), we obtain \(g^s(x) = x\), implying that the period of \(g\) is less than \(2^{|D|}\). This contradicts the assumption that \(g\) is a Hamiltonian cycle and proving that \(G(f)\) is unilaterally connected.
\end{proof}

\Cref{prop717} implies several relevant results. In particular, it establishes that the interaction graph has at most one source component and at most one sink component. The above, guarantees that any pair of components is connected by a path, and ensures that the component graph of \(G(f)\) contains a Hamiltonian path. 

It is conjectured that if \(f\) has a Hamiltonian cycle or an intermediate height dynamics, the component graph of \(G(f)\) is a transitive acylic digraph.


For maximum height, intermediate height, and quasi-Hamiltonian dynamics, the lenght of the attractor plays a crucial role in variable dependence, leading to greater interaction graph connectivity as shown in the following theorem.

\begin{theorem} \label{prop708}
    Let \(f\) be a quasi-Hamiltonian, Hamiltonian of maximum height, or intermediate height with \(p(f)\) being odd. Then \(G(f)\) is strongly connected.
\end{theorem}
\begin{proof}
    Let \(f\) be a Hamiltonian Boolean network of maximum height. By \Cref{prop:prop507}, the digraph \(G(f)\) is connected. By contradiction, assume \(G(f)\) is not strongly connected. Denote \(G(f)[I]\) as a source component of \(G(f)\) induced by \(I \subsetneq V(G(f))\). By \Cref{prop799}, \(f_I\) is a Hamiltonian cycle with period \(p(f_I) = 2^{|I|}\), which contradicts the existence of a fixed point in the dynamics of $f$, thus proving that \(G(f)\) is strongly connected.
    
    Next, consider \(f\) as a quasi-Hamiltonian or intermediate height network with an odd period. By contradiction, assume \(G(f)\) is not strongly connected, and let \(I \subsetneq [n]\) be the set inducing a source component in \(G(f)\). Then, \(f_I\) is a Hamiltonian cycle by \Cref{prop799}.
    
    Let \(y \in \{0,1\}^n\) be an arbitrary configuration. For all \(a \in \mathbb{N}\), it holds that \((f_I)^{a \cdot p(f_I)}(y_I) = y_I\).
    For any periodic configuration \(x \in \{0,1\}^n\) of \(f\), it follows that \(x = f^{p(f)}(x)\). 
    Projecting onto the components in \(I\) gives (\ref{eq:00040}).
    \begin{equation} \label{eq:00040}
        x_I = f^{p(f)}(x)_I = (f_I)^{p(f)}(x_I).
    \end{equation}
    The equality implies that \(p(f)\) is a multiple of \(p(f_I)\). However, since \(p(f_I) = 2^{|I|}\), \(p(f)\) must be even. This contradicts the assumption that \(p(f)\) is odd, proving that \(G(f)\) is strongly connected.
\end{proof}

The Hamiltonian BNs satisfy \(p(f) + h(f) = 2^n\), which allows \Cref{prop708} to be reformulated in terms of height.

\section{Hamiltonian unate Boolean networks} \label{chapter41}

Since not all families of BNs can exhibit Hamiltonian behaviors, this section focuses on the construction of Hamiltonian-type unate Boolean networks.

From the literature, it is known that monotone networks without constant local activation functions have at least two fixed points: \(\vec{0}\) and \(\vec{1}\). Additionally, the length of their limit cycles \(|C| \in \mathbb{N}\) cannot reach \(2^n\), as it is upper-bounded by the number of incomparable vectors with \(n\) components \cite{ARACENA2004237,Sperner}. This bound is presented in (\ref{cotateo}):
\begin{align} \label{cotateo}
    |C| \leq \binom{n}{ \lfloor \frac{n}{2} \rfloor } < 2^n.
\end{align}

On the other hand, conjunctive and disjunctive networks have a height upper-bounded by \(h(f) \leq 2n^2 - 3n + 2\) \cite{yuntive}. If the interaction graph is strongly connected, these networks also have at least two fixed points. 
Additionally, Hamiltonian cycle BNs are balanced in each of their local activation functions, a property not satisfied by certain types of networks, such as conjunctive, disjunctive, or canalizing networks with more than one variable.

In contrast, \cite{SATO1982201} states that Hamiltonian cycle BNs can be modeled using neural networks, a subfamily of unate Boolean networks, and demonstrates that self-duality is a necessary condition to describe bijective behaviors in neural networks. 

Along these lines, we will prove the existence of unate BNs that describe Hamiltonian dynamics of maximum height, intermediate height, and Hamiltonian cycle.

\subsection{Self-dual Boolean networks}

In this section, we explore self-duality, formally described in \Cref{defin0009}. It is conjectured that self-duality constitutes a necessary condition for a Hamiltonian cycle BN to be unate, derived from the properties it induces and empirical results in the literature \cite{Symmetry}.

\begin{definition}\label{defin0009}
    A Boolean network \(f\) with \(n \in \mathbb{N}\) variables is said to be \textbf{self-dual in} \(I \subseteq [n]\), with \(I \neq \emptyset\), if for any configuration \(x \in \{0,1\}^n\), it holds that \(f(x) = \overline{f(\overline{x}^{I})}^{I}\).In the case that $I=[n]$ we say that f is simply {\bf self-dual}.
\end{definition}

\begin{lemma} \label{prop711}
    Let \(I \subseteq [n]\) be non-empty. A Boolean network \(f\) is self-dual in \(I\) if and only if, for any \(x \in \{0,1\}^n\) and \(k \in \mathbb{N}\), it holds that:
    \[
    f^k(x) = \overline{f^k(\overline{x}^I)}^I.
    \]
\end{lemma}
\begin{proof}
    Let \(f\) be self-dual in \(I\). As shown in (\ref{eq:0011}), applying the definition of self-duality twice proves the case \(k=2\).
    \begin{align} \label{eq:0011}
        f^2(x) = f(f(x)) 
        = f(\overline{f(\overline{x}^I)}^I) 
        = \overline{f(f(\overline{x}^I))}^I 
        = \overline{f^2(\overline{x}^I)}^I.
    \end{align}
    Proceed by induction on \(k \in \mathbb{N}\). Suppose that for all \(x \in \{0,1\}^n\), \(f^k(x) = \overline{f^k(\overline{x}^I)}^I\). We prove the case for \(k+1\):
    \begin{align} \label{eq:0012}
        f^{k+1}(x) = f(f^k(x)) 
        = f(\overline{f^k(\overline{x}^I)}^I) 
        = \overline{f(f^k(\overline{x}^I))}^I 
        = \overline{f^{k+1}(\overline{x}^I)}^I.
    \end{align}
    For the reverse implication, consider \(k=1\), which corresponds to the definition of self-duality. Thus, the property holds for all \(k \in \mathbb{N}\), completing the proof.
\end{proof}

A notable result is the relationship between self-duality and the length of the shortest path between configurations in a Hamiltonian cycle dynamic. This is formally established in the following lemma.

\begin{lemma} \label{prop713}
    Let \(f\) be a Hamiltonian cycle Boolean network with \(n \in \mathbb{N}\) variables and \(I \subseteq [n]\), with \(I \neq \emptyset\). A network \(f\) is self-dual in \(I\) if and only if, for any configuration \(x \in \{0,1\}^n\), it holds that:
    \[
    f^{2^{n-1}}(x) = \overline{x}^I.
    \]
\end{lemma}
\begin{proof}
    Suppose \(f\) is self-dual in \(I\). By contradiction, assume there exists a configuration \(x\) such that \(f^{2^{n-1}}(x) \neq \overline{x}^I\). Because the dynamics is a single cycle, without loss of generality we can assume that \(k < 2^{n-1}\) be the smallest value such that \(f^k(x) = \overline{x}^I\). By \Cref{prop711}, we deduce:
    \begin{align} \label{eq000112}
        f^k(\overline{x}^I) = \overline{f^k(x)}^I = x.
    \end{align}
    This implies that the length of the shortest path between \(\overline{x}^I\) and \(x\) in the dynamics is \(k\). However, this contradicts the fact that \(f\) has a cycle of length \(2^n\). Therefore, for any \(x \in \{0,1\}^n\), it holds that \(f^{2^{n-1}}(x) = \overline{x}^I\).

    Conversely, if every configuration $x \in \{0,1\}^n$ satisfies $f^{2^{n-1}}(x) = \overline{x}^I$, the bijectivity of the network ensures that for any $x = f(y)$, we have:
    \begin{align*}
        f^{2^{n-1}}(x) &= f^{2^{n-1}}(f(y)) = f(f^{2^{n-1}}(y)) = f(\overline{y}^I), \\
        \overline{x}^I &= \overline{f(y)}^I.
    \end{align*}
    It follows that $f(\overline{y}^I) = \overline{f(y)}^I$, which is equivalent to the definition of self-duality.
        
\end{proof}

A relevant result for our study is associated with local dependency. In the case where a configuration \(x \in \{0,1\}^n\) induces an edge in \(G_x(f)\), we can prove that such dependency is also induced by \(\overline{x}\).

\begin{lemma} \label{lema502}
    Any self-dual Boolean network \(f\) in \([n]\) satisfies, for every \(x \in \{0,1\}^n\), the equality:
    \begin{equation*}
        (G_x(f), \sigma_{f}) = (G_{\overline{x}}(f), \sigma_{f}).
    \end{equation*}
\end{lemma}
\begin{proof}
    Let \(i, j \in [n]\) be arbitrary, and suppose that the edge \((i, j)\) belongs to the digraph \(G_x(f)\). 
    By the self-duality of \(f\), the following chain of equivalences holds:
    \begin{align*}
        (i, j) \in A(G_x(f)) 
        &\Leftrightarrow (f_j(x) < f_j(x \oplus e_i)) \lor (f_j(x) > f_j(x \oplus e_i)) \\
        &\Leftrightarrow (\overline{f_j(\overline{x})} < \overline{f_j(\overline{x \oplus e_i})}) 
        \lor (\overline{f_j(\overline{x})} > \overline{f_j(\overline{x \oplus e_i})}) \\
        &\Leftrightarrow (f_j(\overline{x}) > f_j(\overline{x} \oplus e_i)) 
        \lor (f_j(\overline{x}) < f_j(\overline{x} \oplus e_i)) \\
        &\Leftrightarrow (i, j) \in A(G_{\overline{x}}(f)).
    \end{align*}
    Therefore, both local interaction graphs are equal; that is, \(G_x(f) = G_{\overline{x}}(f)\).

    If \(\sigma_x(i,j) \in \{+1, -1\}\), this label implies a value for \(x_i\) in the chain of equivalences, establishing an increasing or decreasing behavior of \(f_j\) with respect to index \(i\). Evaluating at \(\overline{x}\), the value in component \(j\) is inverted, and the inequalities switch from strictly greater to strictly lesser, and vice versa. Hence, \(\sigma_x(i,j) = \sigma_{\overline{x}}(i,j)\).
    Consequently, we prove that \((G_x(f), \sigma_{f}) = (G_{\overline{x}}(f), \sigma_{f})\), completing the proof.
\end{proof}

Next, we prove that the self-duality for Hamiltonian cycle BNs ensures maximum in-degree for such components.

\begin{definition} \label{defin0008}
    Let \(f\) be a Boolean network with \(n\) variables \(i, j \in [n]\), \(a \in \{0,1\}\). We denote \(T(f_j, x_i = a) = \{x \in T(f_j) : x_i = a\}\) as a subset of the set of true points of \(f_j\) consisting of configurations such that \(x_i\) takes the value \(a\). Similarly, \(F(f_j, x_i = a) = \{x \in F(f_j) : x_i = a\}\) is defined for the set of false points.
\end{definition}

\begin{lemma} \label{prop453}
    Let \(f\) be a Boolean network with \(n \geq 3\) variables, \(i, j \in [n]\), and suppose that the set of true points \(T(f_j)\) is a multiple of four and non-empty. If there exists some \(a \in \{0,1\}\) such that \(|T(f_j,x_i = a)|\) is odd, then \(f_j\) depends on all variables.
\end{lemma}
\begin{proof}
    First, let us prove that if the cardinalities of \(T(f_j,x_i = 0)\) and \(T(f_j,x_i = 1)\) differ, this implies that \(i \in N^-(j)\). By contraposition, suppose that \(f_j\) does not depend on \(x_i\). In this case, for any configuration \(x \in T(f_j)\), it also holds that \(x \oplus e_i \in T(f_j)\). However, this implies that the cardinalities of \(T(f_j,x_i = 0)\) and \(T(f_j,x_i = 1)\) are equal.

    On the other hand, if there exists \(a \in \{0,1\}\) such that \(|T(f_j,x_i = a)|\) is odd, we can prove that \(N^-(j) \supseteq [n] \setminus \{i\}\). By contradiction, assume that there exists \(k \in [n] \setminus \{i\}\) such that \(k \notin N^-(j)\). This implies that for any configuration \(x \in T(f_j,x_i = a)\), it also holds that \(x \oplus e_k \in T(f_j,x_i = a)\), which contradicts the odd cardinality of \(|T(f_j,x_i = a)|\).

    Finally, since \(|T(f_j)|\) is a multiple of four, the cardinalities of \(T(f_j,x_i = 0)\) and \(T(f_j,x_i = 1)\) must both be odd and different, as their sum must be a multiple of four. If these cardinalities were equal, their sum would be a multiple of two but not of four. This proves that \(N^-(j) = [n]\), implying that \(f_j\) depends on all variables.
\end{proof}

\begin{proposition} \label{prop714}
    If a Hamiltonian cycle Boolean network \(f\) with \(n \in \mathbb{N} \setminus \{2\}\) variables is self-dual in \(I \subseteq [n]\), 
    then, for any \(j \in I\), it holds that \(d^-(j) = n\).
\end{proposition}
\begin{proof}
    Let \(j \in I\) be arbitrary, and denote by \(P\) and \(Q\) the paths in \(\Gamma(f)\) from \(\vec{0} \in \{0,1\}^n\) to the configuration \(\vec{0} \oplus e_I \in \{0,1\}^n\) and vice versa, respectively.
    From \Cref{prop713}, we know that the shortest path length between \(\vec{0}\) and \(\vec{0} \oplus e_I\) is \(2^{n-1}\) implying that \(|P| = |Q| = 2^{n-1}\). Moreover, since the network is a Hamiltonian cycle, these paths cover all edges in the dynamics.

    The path \(P\) must contain at least one configuration \(z \in T(f_j,x_j = 0)\) to transition the value of \(j\) from \(0\) to \(1\). In particular, if \(P\) contains \(k \in \mathbb{N}\) configurations in \(T(f_j,x_j = 0)\), then \(P\) contains exactly \(k - 1\) configurations \(w \in F(f_j,x_j = 1)\), as these transitions are necessary for the change in the value of \(j\) along \(P\).

    If \(w \in F(f_j,x_j = 1)\) is arbitrary, the self-duality of \(f\) implies that \(f_j(w) \neq f_j(\overline{w}^I)\), which in turn establishes that \(w \in F(f_j,x_j = 1)\) if and only if \(\overline{w}^I \in T(f_j,x_j = 0)\). Consequently, \(w \in F(f_j,x_j = 1)\) belongs to the trajectory \(P\) if and only if \(\overline{w}^I \in T(f_j,x_j = 0)\) belongs to the trajectory \(Q\).

    Finally, the cardinality of \(T(f_j,x_j = 0)\) is obtained as the sum of the configurations \(z, \overline{w}^I \in T(f_j,x_j = 0)\), resulting in an odd quantity as described in (\ref{eq:000050}).
    \begin{equation} \label{eq:000050}
        |T(f_j,x_j = 0)| = k + (k - 1) = 2k - 1.
    \end{equation}
    Applying \Cref{prop453}, it is concluded that the index \(j \in I\) has in-degree \(n\) in \(G(f)\), completing the proof.
\end{proof}

\begin{remark}
    It is worth noting that if the Hamiltonian cycle Boolean network is self-dual in \([n]\), by the previous result, its interaction graph is \(K_n\).
\end{remark}

\subsection{Family of Hamiltonian cycle self-dual unate Boolean networks}

Initially, let us analyze the case \(n = 1, 2\). From \Cref{lem6_2}, given \(f\) as a Hamiltonian cycle BN described by a unate Boolean network, it has a unique interaction graph \(G(f)\) with edges labeled with different signs. \Cref{baseej6_1} presents an example of such an interaction graph with signs and details how \(f^{[2]}\) is constructed from the base case \(f^{[1]} = (\overline{x}_1)\), a network of a single variable that is unate and has a Hamiltonian cycle.

\begin{lemma} \label{lem6_2}
    If \(f\) is a Hamiltonian cycle, unate Boolean network with \(n = 2\) variables, then its interaction graph with signs \((G(f), \sigma_f)\) is a cycle without loops, with edges labeled as \(\sigma_f(i, j) = +1\) and \(\sigma_f(j, i) = -1\), where \(i \neq j \in [2]\).
\end{lemma}
\begin{proof}
    An exhaustive analysis is carried out considering all possible configurations of a unate Hamiltonian cycle Boolean network with \(n = 2\) variables.
\end{proof}


\begin{figure}[htb!]
    \begin{subfigure}{.45\textwidth}
    \centering
    \begin{tikzpicture}[scale=1.5,every node/.style={circle, draw=black!60, fill=white, text=black, minimum size=6pt,inner sep=4pt,outer sep=2pt}]
    \node[draw=white, fill=white,text=black](n) at (-35pt,25pt) {$(G(f^{[2]}),\sigma_{f^{[2]}}):$};
    \node (1) at (15pt,15pt) {$1$};
    \node (2) at (-10pt,-10pt) {$2$};
    \path[ultra thick, -to]
    (1) edge[bend right=15,green] (2)
    (2) edge[bend right=15,red] (1)
    ;
    \end{tikzpicture}
    \end{subfigure}
    \begin{subfigure}{.45\textwidth}
        \centering
    \begin{tikzpicture}[scale=1.5,every node/.style={circle, draw=black!60, fill=white, text=black, minimum size=6pt,inner sep=4pt,outer sep=2pt}]
    \node[draw=white, fill=white,text=black](n) at (-45pt,35pt) {$\Gamma(f^{[2]}):$};
    \node (1) at (-15pt,30pt) {$01$};
    \node (2) at (-15pt,0pt) {$00$};
    \node (3) at (15pt,0pt) {$10$};
    \node (4) at (15pt,30pt) {$11$};
    \path[very thick, -to]
    (1) edge[bend right=15,black] (2)
    (2) edge[bend right=15,black] (3)
    (3) edge[bend right=15,black] (4)
    (4) edge[bend right=15,black] (1)
    ;
    \end{tikzpicture}
    \end{subfigure}
    \caption{Interaction graph with signs for \(f^{[2]}(x_1,x_2)=(\overline{x}_2,x_1)\) and its dynamics.}
    \label{baseej6_1}
\end{figure}

To extend this construction to the case \(n+1\), we use \(f^{[n]}\) a Hamiltonian cycle BN with \(n\) variables. 

\begin{definition} \label{construccion}
    Let \(f^{[1]} = (\overline{x}_1)\) be a Boolean network with a single variable. Recursively, networks \(h^{[n+1]}, f^{[n+1]} : \{0,1\}^{n+1} \to \{0,1\}^{n+1}\) are constructed from \(f^{[n]}\), \(n \in \mathbb{N}\). Specifically, if \(z^{[n+1]} \in \{0,1\}^{n+1}\) such that \(f^{[n]}(z_{[n]}) = \vec{0}\) and \(z_{n+1}^{[n+1]} = 0\), then:
    \begin{align*}
        h^{[n+1]}(x) &= (f^{[n]}(x_{[n]}), x_{n+1}), \\
        f^{[n+1]}_i(x)&=(h^{[n+1]}_i(x)\land d_{\overline{z}}(x) ) \lor c_{z}(x)
    \end{align*}
    where $z=z^{[n+1]}$, $i\in [n+1]$ and the Boolean functions \(c_z, d_{\overline{z}} : \{0,1\}^{n+1} \to \{0,1\}\) are conjunctive and disjunctive clauses, defined as \(c_z(z) = 1\), \(d_{\overline{z}}(\overline{z}) = 0\), and take the opposite value otherwise.
\end{definition}

\Cref{fig:h3} shows \(f^{[2]}(x_1, x_2) = (\overline{x}_2, x_1)\) and the associated auxiliary network \(h^{[3]}(x_1, x_2, x_3) = (\overline{x}_2, x_1, x_3)\).

By modifying the auxiliary network, swapping the preimages of \(\vec{0} \in \{0,1\}^{n+1}\) and \(\vec{1} \in \{0,1\}^{n+1}\), a new BN \(f^{[n+1]}\) is constructed, which is a Hamiltonian cycle, unate, and self-dual in \([n+1]\). \Cref{fig:h3ed} illustrates \(f^{[3]}\), a self-dual BN in \([3]\) defined by swapping preimages in the auxiliary network \(h^{[3]}\).

\begin{figure}[htb!]
    \centering
    \begin{tikzpicture}[scale=1.6,every node/.style={circle, draw=black!60, fill=white, text=black, minimum size=6pt,inner sep=4pt,outer sep=2pt}]
    \node[draw=white, fill=white,text=black](n) at (-90pt,40pt) {$\Gamma(h^{[3]}):$};
    \node (1) at (-60pt,30pt) {$010$};
    \node (2) at (-60pt,0pt) {$000$};
    \node (3) at (-30pt,0pt) {$100$};
    \node (4) at (-30pt,30pt) {$110$}; %
    \node (5) at (30pt,30pt) {$011$};
    \node (6) at (30pt,0pt) {$001$};
    \node (7) at (60pt,0pt) {$101$};
    \node (8) at (60pt,30pt) {$111$};
    \path[very thick, -to]
    (1) edge[bend right=15,purple] (2)
    (2) edge[bend right=15,black] (3)
    (3) edge[bend right=15,black] (4)
    (4) edge[bend right=15,black] (1) %
    (5) edge[bend right=15,black] (6)
    (6) edge[bend right=15,black] (7)
    (7) edge[bend right=15,blue] (8)
    (8) edge[bend right=15,black] (5) 
    ;
    \end{tikzpicture}
    \caption{Dynamics of the auxiliary network \(h^{[3]}(x_1,x_2,x_3)=(\overline{x}_2,x_1,x_{3})\).}
    \label{fig:h3}
\end{figure}

\begin{figure}[htb!]
    \centering
    \begin{tikzpicture}[scale=1.5,every node/.style={circle, draw=black!60, fill=white, text=black, minimum size=6pt,inner sep=4pt,outer sep=2pt}]
    \node[draw=white, fill=white,text=black](n) at (-90pt,50pt) {$\Gamma(f^{[3]}):$};
    \node[draw=white, fill=white,text=black](n) at (0pt,50pt) {$c_z$};
    \node[draw=white, fill=white,text=black](n) at (0pt,-19pt) {$d_{\overline{z}}$};
    \node (1) at (-60pt,30pt) {$010$};
    \node (2) at (-60pt,0pt) {$000$};
    \node (3) at (-30pt,0pt) {$100$};
    \node (4) at (-30pt,30pt) {$110$}; %
    \node (5) at (30pt,30pt) {$011$};
    \node (6) at (30pt,0pt) {$001$};
    \node (7) at (60pt,0pt) {$101$};
    \node (8) at (60pt,30pt) {$111$};
    \path[very thick, -to]
    (2) edge[bend right=15,black] (3)
    (3) edge[bend right=15,black] (4)
    (4) edge[bend right=15,black] (1) %
    (5) edge[bend right=15,black] (6)
    (6) edge[bend right=15,black] (7)
    (8) edge[bend right=15,black] (5) 
    (1) edge[bend left=35,purple] (8)
    (7) edge[bend left=35,blue] (2)
    ;
    \end{tikzpicture}
    \caption{Dynamics of the network \(f^{[3]}\) constructed from the auxiliary network \(h\).}
    \label{fig:h3ed}
\end{figure}

\begin{lemma} \label{teocidu}
    For any \(n \in \mathbb{N}\), the Boolean network \(f^{[n]}\) is a Hamiltonian cycle and self-dual in \([n]\).
\end{lemma}
\begin{proof}
    {\bf Base Case \(n = 1, 2, 3\)}: For \(n = 1\), the network \(f^{[1]} = (\overline{x}_1)\) is trivially a Hamiltonian cycle and self-dual in \([1]\). It is explicitly verified that \(f^{[2]}\) and \(f^{[3]}\) are Hamiltonian cycles and self-dual in \([2]\) and \([3]\), respectively, according to the local activation functions:
    \begin{align*}
        f^{[2]}_1(x) &= \overline{x}_2, \\
        f^{[2]}_2(x) &= x_1, \\
        f^{[3]}_1(x) &= (\overline{x}_1 \land \overline{x}_2) \lor (\overline{x}_2 \land \overline{x}_3) \lor (\overline{x}_1 \land \overline{x}_3), \\
        f^{[3]}_2(x) &= (x_1 \land x_2) \lor (x_1 \land \overline{x}_3) \lor (x_2 \land \overline{x}_3), \\
        f^{[3]}_3(x) &= (\overline{x}_1 \land x_3) \lor (x_2 \land x_3) \lor (\overline{x}_1 \land x_2).
    \end{align*}

    {\bf Induction for \(n \geq 3\)}: Assume that \(f^{[n]}\) is a Hamiltonian cycle and self-dual in \([n]\), and prove that \(f^{[n+1]}\) also satisfies this property.

    {\bf Hamiltonian Cycle}: Let \(c_z, d_{\overline{z}}\) be the clauses defined for \(f^{[n+1]}\). Observe that:
    \[
    f^{[n+1]}(z) = \vec{1}, \quad f^{[n+1]}(\overline{z}) = \vec{0}.
    \]
    For \(x \in \{0,1\}^{n+1} \setminus \{z, \overline{z}\}\), the network reduces to:
    \begin{equation} \label{simplif}
        f^{[n+1]}(x) = (f^{[n]}(x_{[n]}), x_{n+1}).
    \end{equation}
    Since \(f^{[n]}\) is a Hamiltonian cycle, the configuration \(\vec{0}_{[n]}\) reaches all configurations \(u \in \{0,1\}^{n+1}\) with \(u_{n+1} = 0\). Similarly, \(\vec{1}_{[n]}\) reaches all configurations \(v \in \{0,1\}^{n+1}\) with \(v_{n+1} = 1\). Finally, since \(f^{[n+1]}(z) = \vec{1}\) and \(f^{[n+1]}(\overline{z}) = \vec{0}\), the dynamics are strongly connected, and \(f^{[n+1]}\) is a Hamiltonian cycle.

    {\bf Self-Duality}: Let \(x \in \{0,1\}^{n+1}\setminus \{z, \overline{z}\}\) be arbitrary, we can use \eqref{simplif} and the inductive hypothesis:
    \begin{align*}
        f^{[n+1]}(x) &= (f^{[n]}(x_{[n]}), x_{n+1}) \\
        &= \overline{(f^{[n]}(\overline{x}_{[n]}) \overline{x}_{n+1})} \\
        &= \overline{f^{[n+1]}(\overline{x})}.
    \end{align*}
    On the other hand, due to the definition of $f^{[n+1]}$ we know that $f^{[n+1]}(z)=\overline{f^{[n+1]}(\overline{z})}$ and we prove self-duality. Hence, \(f^{[n+1]}\) is a Hamiltonian cycle and self-dual.
\end{proof}

A distinctive feature of the networks in \Cref{construccion} is that, except for \(n = 2\), they always generate complete digraphs including loops. This result, derived from the self-dual structure of these networks, is formalized below.

\begin{corollary}
   \(G(f^{[n]}) = K_n\), \(n \in \mathbb{N} \setminus \{2\}\).
\end{corollary}
\begin{proof}
    For \(n = 1\), \(f^{[1]} = (\overline{x}_1)\) and the result is immediate.

    For \(n \geq 3\), \Cref{teocidu} establishes that \(f^{[n]}\) is self-dual in \([n]\). By \Cref{prop714} it is concluded that for all \(j \in [n]\) the in-degree satisfies \(d^-(j) = n\) and implying \(G(f^{[n]}) = K_n\).
\end{proof}

From this construction, certain properties justify the unate nature of the network. Describing \(f^{[n]}_j\), with \(j \in [n]\), as the concatenation of a conjunctive and a disjunctive clause of size \(k \in \{j, j+1, \ldots, n\}\) over the variable \(x_j\), it is possible to infer the value of \(f^{[n]}_j\) when evaluated at an arbitrary configuration. This requires projecting the first \(k\) variables of the configuration to be evaluated. By a case analysis, given \(z^{[n]} \in \{0,1\}^n\) such that \(f^{[n]}(z^{[n]}) = \vec{1}\), the cases where there exist \(i, j \in [n]\) satisfying the inequality \(f^{[n]}_j(z^{[n]}) \neq f^{[n]}_j(z^{[n]} \oplus e_i)\) are summarized in \Cref{table0001}.

\begin{table}[htb!]
    \centering
    \begin{tabular}{|p{1.2cm} |p{1.2cm}|p{1.2cm}|p{1.2cm}|p{1.2cm}|p{1.2cm}|p{1.2cm}|}
        \hline
        &\multicolumn{2}{c|}{$i>j$} & \multicolumn{2}{c|}{$i=j$} & \multicolumn{2}{c|}{$i<j$} \\
        \hline
        &$i$ even&$i$ odd
        &$i$ even&$i$ odd
        &$j$ even&$j$ odd \\
        \hline
        n even
        &True&False&False &True&True&False \\
        \hline
        n odd
        &False&True&True &False&False&True \\
        \hline
    \end{tabular}
    \caption{Summary of cases where there exist \(i, j \in [n]\) such that \(f^{[n]}_j(z^{[n]}) \neq f^{[n]}_j(z^{[n]} \oplus e_i)\).}
    \label{table0001}
\end{table}

By the definition of $f^{[n]}$, we have $f^{[n]}_i(x)=(h^{[n]}_i(x) \land d_{\overline{z}^{[n]}}(x)) \lor c_{z^{[n]}}(x)$. Since $c_{z^{[n]}}(z^{[n]}) = 1$ and $d_{\overline{z}^{[n]}}(\overline{z}^{[n]}) = 0$, 
$f^{[n]}(z^{[n]}) = 1$ and $f^{[n]}(\overline{z}^{[n]}) = 0$.
 Such behavior extends to $f^{[k]}$ for $k \in \{j, j+1, \ldots, n-1\}$ when evaluated at $z^{[k]}$ and $\overline{z}^{[k]}$. Following this approach, we generalize \Cref{table0001} to arbitrary configurations.

\begin{definition}
    For \(x \in \{0,1\}^n\), \(k_x \in \{1,2,...,n\}\) denotes the largest value such that \(x_{[k_x]} = (0,1,0,\ldots)\) or \(x_{[k_x]} = (1,0,1,\ldots)\).
\end{definition}

We denote $t_k\in\{0,1\}^k$ as either $\overline{0}^E$ or $\overline{1}^E$. The recursive construction of \(f^{[n]}\) incorporates clauses defined from the configuration \(z^{[n]}\) and the following result proves that \(z^{[n]} = t_n\), contributing to understanding the construction of \(f^{[n]}\).

\begin{lemma} \label{lem0015}
    The configurations \(z^{[n]}\) and \(\overline{z}^{[n]}\) from \Cref{construccion} are of type \(t_k\).
    This establishes conjunctive clauses \(c_{t_k}\) and disjunctive clauses \(d_{\overline{t_k}}\) at each recursive step of the construction.
\end{lemma}
\begin{proof}
    For the base case \(n = 2\), these clauses are induced from \(z^{[2]} = (1,0)\) and its negation.
    Assume by induction on \(n \geq 2\) that for all \(k \leq n\), \(z^{[k]} = t_k\). By construction, the last component of any \(z^{[k]}\) has value zero, and by the induction hypothesis \(z^{[n]} = t_n = (t_{n-1}, 0)\). Given that \(f^{[n+1]}(\overline{z}^{[n]}, 0) = \vec{1}\), it follows that \(z^{[n+1]} = (\overline{z}^{[n]}, 0) = (\overline{t_n}, 0) = t_{n+1}\), completing the proof.
\end{proof}

\Cref{lem0015} establishes a periodic structure of \(z^{[n]}\) and \(\overline{z}^{[n]}\) as $n$ grows and gives first insights into predicting the value of \(f_j(x)\). We will proceed to study the effect of the change in a component $i\in[n]$ of \(f_j(x\oplus e_i)\) in order to establish a new tool.

\begin{lemma} \label{lem0014}
    Given \(f^{[n]}\) described in \Cref{construccion}, \(i \in [n]\), and \(x \in \{0,1\}^n\) arbitrary:
    \begin{enumerate}
        \item For all \(i \in [n] \setminus \{j\}\) such that \(k_x, k_{x \oplus e_i} < j\), it holds that \(f^{[n]}_j(x) = f^{[n]}_j(x \oplus e_i)\), \label{inc00001}
        \item \(k_x, k_{x \oplus e_i} < i\) if and only if \(k_x = k_{x \oplus e_i}\), \label{inc00002}
        \item if \(k_x = k_{x \oplus e_i}\), then for each \(j \in [n] \setminus \{i\}\), it holds that \(f^{[n]}_j(x) = f^{[n]}_j(x \oplus e_i)\). \label{inc00003}
    \end{enumerate}
\end{lemma}
\begin{proof}
    (\ref{inc00001}) Suppose \(k_x, k_{x \oplus e_i} < j\), noting that the conjunctive and disjunctive clauses in \(f^{[n]}_j\) do not contribute to the evaluation of \(x\) and \(x \oplus e_i\), as these clauses are of size \(j\) and above. Hence, the value of \(f^{[n]}_j\) depends only on the variable \(x_j\), implying:
    \[
    f^{[n]}_j(x \oplus e_i) = (x \oplus e_i)_j = x_j = f^{[n]}_j(x).
    \]

    (\ref{inc00002}.1) If \(k_x, k_{x \oplus e_i} < i\), note that \(x\) and \(x \oplus e_i\) differ only in component \(i\). If \(k_x \neq k_{x \oplus e_i}\), this implies there exists \(j < i\) where \(x\) and \(x \oplus e_i\) differ, leading to a contradiction. Thus, \(k_x = k_{x \oplus e_i}\).

    (\ref{inc00002}.2) Assume \(k_x = k_{x \oplus e_i}\). If the change in component \(i\) of \(x\) does not alter the value of \(k_x\), then \(k_x, k_{x \oplus e_i} < i\), as desired.

    (\ref{inc00003}) Let \(k_x = k_{x \oplus e_i}\) and \(j \in [n] \setminus \{i\}\). 
    
    If \(j > k_x\), from (\ref{inc00001}) it follows that \(f^{[n]}_j(x) = f^{[n]}_j(x \oplus e_i)\).
    
    If \(j \leq k_x\), the value of \(f^{[n]}_j(x)\) is determined by \(x_j\) or the conjunctive and disjunctive clauses. By \Cref{lem0015}, these clauses are of type \(c_{t_k}\) and \(d_{\overline{t_k}}\), with \(k \geq j\). Since \(k_x = k_{x \oplus e_i}\), it holds that \(x_j = (x \oplus e_i)_j\), and thus the clauses of \(f^{[n]}_j(x)\) and \(f^{[n]}_j(x \oplus e_i)\) take the same values:
    \[
        c_{t_k}(x) = c_{t_k}(x \oplus e_i), \quad d_{\overline{t_k}}(x) = d_{\overline{t_k}}(x \oplus e_i).
    \]
    Therefore, \(f^{[n]}_j(x) = f^{[n]}_j(x \oplus e_i)\).
\end{proof}

\Cref{lem0014} is further refined considering the parity of \(n\) and the parity of the altered index in \(x \in \{0,1\}^n\). \Cref{lem0016} establishes how the conjunctive and disjunctive clauses interact with the size of the oscillating configuration \(x_{[k_x]}\) for predicting the evaluation of the local activation function.

\begin{lemma} \label{lem0016}
    For the Boolean network \(f^{[n]}\), \(j \in [n]\), and \(x \in \{0,1\}^n\) such that \(k_x \geq j\), it holds that:
    \begin{enumerate}
        \item Assuming \(n\) is even and \(x_{[k_x]} = z^{[k_x]}\), or \(n\) is odd and \(x_{[k_x]} = \overline{z}^{[k_x]}\),
        then \(f^{[n]}_j(x \oplus e_p) = 0\) and \(f^{[n]}_j(x \oplus e_q) = 1\). \label{eq:00006}
        \item Conversely, if \(n\) is odd and \(x_{[k_x]} = z^{[k_x]}\), or \(n\) is even and \(x_{[k_x]} = \overline{z}^{[k_x]}\), then \(f^{[n]}_j(x \oplus e_p) = 1\) and \(f^{[n]}_j(x \oplus e_q) = 0\). \label{eq:00007}
    \end{enumerate}
    Where \(p, q \in \{j+1, \ldots, k_x\}\) such that \(p\) is even and \(q\) is odd.
\end{lemma}
\begin{proof}
    Let \(i \in \{j+1, \ldots, k_x\}\) and note that the parity or oddness of \(k_{x \oplus e_i} = i - 1\) depends on \(i\).
    From the proof of \Cref{lem0015}, and denoting \(z = z^{[n]}\), \(f^{[n]}_j\) includes conjunctive clauses \(c_{z_{[n]}}\), \(c_{\overline{z}_{[n-1]}}\), \(c_{z_{[n-2]}}\), alternating in negation up to index \(j\). Similarly, the disjunctive clauses in \(f^{[n]}_j\) alternate as \(d_{\overline{z}_{[n]}}\), \(d_{z_{[n-1]}}\), \(d_{\overline{z}_{[n-2]}}\), also up to index \(j\). Denoting \(p, q \in \{j+1, \ldots, k_x\}\) as described in the statement:

    (\ref{eq:00006}) If \(n\) is even and \(x_{[k_x]} = z_{[k_x]}\), this implies that \(f^{[n]}_j(x \oplus e_p) = d_{z_{[p-1]}}(z) = 0\).
    This result follows since \((x \oplus e_p)_{[p-1]} = x_{[p-1]} = z_{[p-1]}\), and because \(n\) is even, the clause of size \(p-1\) (odd) activated by evaluating \(z_{[p-1]}\) corresponds to \(d_{z_{[p-1]}}\).
    Similarly, for \(n\) odd and \(x_{[k_x]} = \overline{z}_{[k_x]}\), we deduce that
    \(f^{[n]}_j(x \oplus e_q) = c_{z_{[q-1]}}(z) = 1\).

    (\ref{eq:00007}) Since \(f^{[n]}\) is self-dual in \([n]\), applying (1) yields the desired result.
\end{proof}

The results obtained are sufficient to demonstrate that the family of BNs \(f^{[n]}\) is unate.

\begin{theorem} \label{teo0014}
    The Boolean network $f^{[n]}$ is unate.
\end{theorem}
\begin{proof}
    Suppose, by contradiction, that \(f^{[n+1]}\) is not a unate Boolean network.
    Then, there exist \(i, j \in [n]\), \(x \in \{0,1\}^n \setminus \{z^{[n]}, z^{[n]} \oplus e_i\}\), such that \(x_i = z^{[n]}_i\), and without loss of generality, satisfies (\ref{eq00031}) and (\ref{eq00032}).
    \begin{align} \label{eq00031}
        f^{[n]}_j(z^{[n]}) &> f^{[n]}_j(z^{[n]} \oplus e_i), \\
        f^{[n]}_j(x) &< f^{[n]}_j(x \oplus e_i). \label{eq00032}
    \end{align}
    Evaluating at \(z^{[n]}\) is justified because, from \Cref{lema502}, the local interaction graphs with signs for \(z\) and \(\overline{z}\) coincide. Moreover, these are the only configurations with distinct images in the unate Boolean network \(h = (f^{[n-1]}, x_n)\).

    Note that the case \(k_x, k_{x \oplus e_i} < i\) cannot occur, as it would contradict (\ref{eq00032}) by \Cref{lem0014}. Assuming \(n\) is even, we proceed with a case analysis on the relationship between \(i\) and \(j\):

    $\bullet$ {\bf Case \(i > j\):} From \Cref{table0001}, it follows that \(i\) is even.
    Suppose \(k_x > i\), so \(k_{x \oplus e_i} = i-1\), an odd value.
    However, \Cref{lem0016} implies that \(f^{[n]}_j(x \oplus e_i) = 0\), contradicting our assumption.
    Conversely, if \(k_{x \oplus e_i} > i\), it follows that \(k_x = i-1\) (odd).
    Since \(x_i = z^{[n]}_i\), the only possible case is \(x_{[i-1]} = \overline{z}^{[i-1]}\), which returns 1 when evaluated in the network by \Cref{lem0016}. This contradicts (\ref{eq00032}), ruling out the case \(i > j\).

    $\bullet$ {\bf Case \(i = j\):} Since \(n\) is even, \Cref{table0001} implies that \(i\) is odd.

    Suppose \(k_x > i\) and note that \(k_{x \oplus e_i} = i-1\), an even value.
    However, since \(i = j\), \(k_{x \oplus e_i} < j\), implying that \(f^{[n]}_j(x \oplus e_i) = (x \oplus e_j)_j = \overline{x}_j\) and \(f^{[n]}_j(z^{[n]} \oplus e_i) = (z^{[n]} \oplus e_j)_j = \overline{z}^{[n]}_j\), values known to be equal because \(x_i = z^{[n]}_i\). This contradicts (\ref{eq00031}) and (\ref{eq00032}).

    Suppose \(k_{x \oplus e_i} > i\), so \(k_x = i-1\), an even value.
    Since \(i = j\), \(k_x < j\), implying \(f^{[n]}_j(x) = x_j\) and \(f^{[n]}_j(z^{[n]}) = z^{[n]}_j\).
    These values are equal because \(x_i = z^{[n]}_i\), contradicting (\ref{eq00031}) and (\ref{eq00032}).
    Therefore, the case \(i = j\) is impossible.

    $\bullet$ {\bf Case \(i < j\):} From \Cref{table0001}, \(j\) is even.
    Since \(k_{z^{[n]} \oplus e_i} = i-1 < j\), \(f^{[n]}_j(z^{[n]} \oplus e_i) = z^{[n]}_j\). Furthermore, inequality (\ref{eq00031}) implies \(z^{[n]}_j = 0\).

    Analyzing \(j\), \Cref{lem0014} precludes the possibility of \(k_x, k_{x \oplus e_i} < j\), as this would contradict (\ref{eq00032}).

    If \(k_x > j\), then \(k_{x \oplus e_i} = i-1 < j\), so \(f^{[n]}_j(x \oplus e_i) = (x \oplus e_i)_j = x_j\), and from (\ref{eq00032}), \(x_j = 1\).
    Since \(x_i = z^{[n]}_i\), it follows that \(x_{[k_x]} = z^{[n]}_{[k_x]}\), which contradicts \(x_j \neq z^{[n]}_j = 0\).

    Suppose \(k_{x \oplus e_i} > j\), then \(k_x = i-1\).
    Since \(x_i = z^{[n]}_i\), it follows that \(x_{[i-1]} = \overline{z}^{[n]}_{[i-1]}\) and \((x \oplus e_i)_{[k_{x \oplus e_i}]} = \overline{z}^{[n]}_{[k_{x \oplus e_i}]}\).
    Noting that \(i < j < k_{x \oplus e_i}\), \((x \oplus e_i)_j = x_j = \overline{z}^{[n]}_j = 1\), and since \(k_x = i-1 < j\), \(f^{[n]}_j(x) = x_j = 0\), contradicting our assumption.

    The case for \(n\) odd follows similarly to the even case. It is proven that \(f^{[n]}_j\) is unate, and therefore, the network \(f\) is unate.
\end{proof}

It is worth noting that the family described above is a non-neural network.  
To establish this, it suffices to show that one of its local activation functions is not a threshold function.  
From \cite{doi:10.1137/1.9780898718539} we know that a function is non-threshold whenever it is not assumable.

\begin{definition}
    A Boolean function $f:\{0,1\}^{n}\to\{0,1\}$ is said to be \textbf{assumable} if, for every sequence $x^{1},x^{2},\dots,x^{k}\in T(f)$ and every sequence $y^{1},y^{2},\dots,y^{k}\in F(f)$, one has
    \[
    \sum_{i=1}^{k}x^{i}\neq \sum_{i=1}^{k}y^{i}.
    \]
\end{definition}

\begin{proposition}
    The Boolean networks \(f^{[n]}\), $n\geq 4$, are non-neural networks.
\end{proposition}
\begin{proof}
    We can prove that $f^{[n]}_{n}$ is a non-threshold function.  
    By the definition of the family, we know that for every $x \in \qn\setminus\{z^{[n]},\overline{z}^{[n]}\}$ one has $f^{[n]}_{n}(x)=x_{n}$.  
    We proceed by distinguishing two cases according to the parity of $n$.
    
    $\bullet$ Suppose first that $n$ is even. Recall that $\overline{z}^{[n]}\in F(f^{[n]}_{n})$ has the form
    \[
    \overline{z}^{[n]}=\sum_{i=1}^{\frac{n}{2}} e_{2i}.
    \]
    Note that $e_{1}$ and $\overline{z}^{[n]}\oplus e_{1}$ are both different from $z^{[n]}$ and $\overline{z}^{[n]}$.  
    By the definition of $f^{[n]}_{n}$ this yields $e_{1}\in F(f^{[n]}_{n})$ and $\overline{z}^{[n]}\oplus e_{1}\in T(f^{[n]}_{n})$.  
    But this implies that 
    \[
    \overline{z}^{[n]}\oplus e_{1}=\overline{z}^{[n]}+e_{1},
    \]
    which shows that $f^{[n]}_{n}$ is a non-threshold function.
    
    $\bullet$ In the case where $n$ is odd, by the construction of the family we know that $z^{[n]}\in T(f^{[n]}_{n})$ has the form
    \[
    z^{[n]}=\sum_{i=1}^{\frac{n-1}{2}} e_{2i}.
    \]
    Now, the configurations $e_{2}$ and $z^{[n]}\oplus e_{2}$ are distinct from both $z^{[n]}$ and $\overline{z}^{[n]}$, which by definition gives $e_{2}\in F(f^{[n]}_{n})$ and $z^{[n]}\oplus e_{2}\in F(f^{[n]}_{n})$.  
    However, this implies that 
    \[
    (z^{[n]}\oplus e_{2})+e_{2}=z^{[n]},
    \]
    and hence $f^{[n]}_{n}$ is a non-threshold function.
\end{proof}

On the other hand, a corollary of \Cref{{teo0014}} extending this intermediate result between the auxiliary network \(h^{[n]}\) and \(f^{[n]}\) is presented below.

\begin{corollary} \label{cor00031}
    The Boolean networks \(h^{[n]} \dot{\lor} c_{z^{[n]}}\) and \(h^{[n]} \dot{\land} d_{\overline{z}^{[n]}}\), defined in (\ref{eq:000235}) and (\ref{eq:000236}) for all $j\in [n]$, are unate.
    \begin{align}
        (h^{[n]} \dot{\lor} c_{z^{[n]}})_j(x)&=(h^{[n]}_j \lor c_{z^{[n]}})(x), \label{eq:000235} \\
        (h^{[n]} \dot{\land} d_{\overline{z}^{[n]}})_j(x)&=(h^{[n]}_j \land d_{\overline{z}^{[n]}})(x). \label{eq:000236}
    \end{align}
\end{corollary}
\begin{proof}
    To prove that \(h^{[n]} \dot{\lor} c_{z^{[n]}}\) is unate, we proceed analogously to the proof of \Cref{teo0014}. This is because the fact that \(f^{[n+1]}(\overline{z}) = \vec{0}\) is not used.

    For the network \(h^{[n]} \dot{\land} d_{\overline{z}^{[n]}}\), from \Cref{lema502}, the local interaction graphs with signs satisfy:
    \[
    (G_{\overline{z}}(h^{[n]} \dot{\land} d_{\overline{z}^{[n]}}), \sigma_{g_1}) = (G_{\overline{z}}(f^{[n+1]}), \sigma_{g_2}) = (G_z(f^{[n+1]}), \sigma_{g_3}) = (G_z(h^{[n]} \dot{\lor} c_{z^{[n]}}), \sigma_{g_4}).
    \]
    Hence, \(h^{[n]} \dot{\land} d_{\overline{z}^{[n]}}\) is a unate network.
\end{proof}

\subsection{General case of Hamiltonian Boolean networks}

Unate BNs can model Hamiltonian behaviors because, from \Cref{cor00031}, we can manipulate the local activation functions to establish any Hamiltonian dynamics of maximum and intermediate height.

\begin{theorem} \label{teo:000100}
    Every Hamiltonian digraph \(G_\Gamma \in \mathcal{G}(n)\) has an associated unate Boolean network with dynamics isomorphic to \(G_\Gamma\).
\end{theorem}
\begin{proof}
    Let \(n \in \mathbb{N}\) and \(G_\Gamma\) be a Hamiltonian digraph with \(2^n\) vertices.
    The case where \(G_\Gamma\) is a Hamiltonian cycle follows from \Cref{teo0014}.
    If \(G_\Gamma\) is not a Hamiltonian cycle, we can define \(g^{[n]} \in \mathcal{F}(G_\Gamma)\) from \(f^{[n]}\).
    To do so, it suffices to change the arc \((z^{[n]}, \vec{1}) \in A(\Gamma(f^{[n]}))\) to the arc \((z^{[n]}, u)\), \(u \neq \vec{1}\), describing a Hamiltonian dynamic, but not of the type of a Hamiltonian cycle.

    Since by definition of \(f^{[n]}\) we can define \(g^{[n]} \in \mathcal{F}(G_\Gamma)\) as described in (\ref{eq:00701}).

    \begin{equation} \label{eq:00701}
        g^{[n]}_j(x) =
        \begin{cases}
            (h^{[n]}_j \dot{\land} d_{\overline{z}^{[n]}})(x) & \text{if } x = z^{[n]} \text{ and } u_j = 0, \\
            f^{[n]}_j(x) & \text{otherwise}.
        \end{cases}
    \end{equation}

    According to \Cref{teo0014} and \Cref{cor00031}, the local activation functions \(g^{[n]}_j\) are unate.
    If \(x \in \{0,1\}^n \setminus \{z^{[n]}\}\), it follows that \(g^{[n]}(x) = f^{[n]}(x)\), forming a Hamiltonian Boolean network.
    By definition, \(g^{[n]}(z^{[n]}) = u\), proving that \(g^{[n]}\) is a unate and Hamiltonian Boolean network of maximum height when \(u = z^{[n]}\), or of intermediate height otherwise.
\end{proof}

The network \(h^{[n]}\) served as an auxiliary tool for constructing Hamiltonian cycle dynamics.
However, this construction can be exploited further to extend the implications of \Cref{cor00031}.

Transitioning from a Hamiltonian cycle to another Hamiltonian dynamic requires changing only one arc.
However, we demonstrate that this can be done for both the image of \(z^{[n]}\) and \(\overline{z}^{[n]}\).

\begin{definition} \label{defi0019}
    A directed graph \(G_\Gamma \in \mathcal{G}(n)\) is called 2-Hamiltonian if all arcs of the digraph can be covered by two trajectories of length \(2^{n-1}\).
\end{definition}

\begin{figure}[htb!]
    \centering
    \begin{tikzpicture}[scale=1,every node/.style={circle, draw=black!60, fill=white, text=black, minimum size=6pt,inner sep=4pt,outer sep=4pt}]
    \node[draw=white, fill=white, text=black](n) at (-110pt,35pt) {$G_{\Gamma}:$};
    \node (1) at (-80pt,20pt) {$1$};
    \node (2) at (-35pt,-25pt) {$2$};
    \node (3) at (-15pt,30pt) {$3$};
    \node (4) at (25pt,-40pt) {$4$};
    \node (5) at (45pt,30pt) {$5$};
    \node (6) at (80pt,-25pt) {$6$};
    \node (7) at (110pt,20pt) {$7$};
    \node (8) at (140pt,-25pt) {$8$};
    \path[ultra thick, -to]
    (1) edge[bend right=10] (2) 
    (2) edge[bend right=10] (4) 
    (3.-160.00) edge[controls=+(-190.00:30pt)and +(-100.00:30pt),black] (3.-120.00)
    (4) edge[bend right=10] (6) 
    (6) edge[bend right=10] (5)
    (5) edge[bend right=10] (3)
    (8) edge[bend right=10] (7)
    (7) edge[bend right=10] (5);
    \end{tikzpicture}
    \caption{Example of a 2-Hamiltonian digraph.}
    \label{fig:figu11}
\end{figure}

2-Hamiltonian digraphs illustrate the ability to modify two images of the auxiliary network \(h^{[n]}\) while maintaining the property of being a unate Boolean network.
Examples of 2-Hamiltonian digraphs include Hamiltonian digraphs, \(\Gamma(h^{[n]})\), or the one described in \Cref{fig:figu11}, among others. For this last example, the arcs can be covered by two trajectories \(P:1,2,4,6,5\) and \(Q:8,7,5,3,3\), both of equal length, demonstrating its 2-Hamiltonian property.

Note that 2-Hamiltonian digraphs do not necessarily induce properties in the connectivity of the interaction graph.
A clear example of a disconnected interaction graph is \(G(h^{[n]})\).

\begin{corollary}
    Any 2-Hamiltonian digraph \(G_\Gamma \in \mathcal{G}(n)\) has a unate Boolean network with dynamics isomorphic to \(G_\Gamma\).
\end{corollary}
\begin{proof}
    Let \(n \in \mathbb{N}\), and suppose \(G_\Gamma\) is 2-Hamiltonian, distinct from a Hamiltonian cycle, with \(2^n\) vertices.
    Since \(f^{[n]}\) is a Hamiltonian cycle, it is also 2-Hamiltonian, with trajectories \(P\) and \(Q\) defined in (\ref{eq:00710}) and (\ref{eq:00711}), respectively.
    \begin{align} \label{eq:00710}
        P &= \vec{1}, f^{[n]}(\vec{1}), (f^{[n]})^2(\vec{1}), \ldots, (f^{[n]})^{2^{n-1}-2}(\vec{1}), \overline{z}^{[n]}, \\
        Q &= \vec{0}, f^{[n]}(\vec{0}), (f^{[n]})^2(\vec{0}), \ldots, (f^{[n]})^{2^{n-1}-2}(\vec{0}), z^{[n]}. \label{eq:00711}
    \end{align}
    Since \(G_\Gamma\) is 2-Hamiltonian, it can be covered using \(P\) and \(Q\).
    Let \(u, v \in \{0,1\}^n\) be the images of \(\overline{z}^{[n]}\) and \(z^{[n]}\) in the coverage of \(G_\Gamma\), and define \(g^{[n]} \in \mathcal{F}(G_\Gamma)\) as the network describing this coverage.
    Based on \(f^{[n]}\), the arcs \((z^{[n]}, \vec{1})\) and \((\overline{z}^{[n]}, \vec{0}) \in A(\Gamma(f^{[n]}))\) are replaced with \((z^{[n]}, u)\) and \((\overline{z}^{[n]}, v)\).
    \(g^{[n]}\) is described as shown in (\ref{eq:00700}).
    
    \begin{equation} \label{eq:00700}
        g^{[n]}_j(x) =
        \begin{cases}
            (h^{[n]}_j \dot{\land} d_{\overline{z}^{[n]}})(x) & \text{if } x = z^{[n]} \text{ and } u_j = 0, \\
            (h^{[n]}_j \dot{\lor} c_{z^{[n]}})(x) & \text{if } x = \overline{z}^{[n]} \text{ and } v_j = 1, \\
            f^{[n]}_j(x) & \text{otherwise}.
        \end{cases}
    \end{equation}
    
    From \Cref{teo0014} and \Cref{cor00031}, \(g^{[n]}_j\) are unate local activation functions.
    For \(x \in \{0,1\}^n \setminus \{z^{[n]}, \overline{z}^{[n]}\}\), \(g^{[n]}(x) = f^{[n]}(x)\).
    Additionally, \(g^{[n]}(z^{[n]}) = u\), \(g^{[n]}(\overline{z}^{[n]}) = v\), and it follows that \(g^{[n]} \in \mathcal{F}(G_\Gamma)\).
\end{proof}

\section{Conclusions} \label{chapter51}

In this work, Hamiltonian dynamics were addressed with the aim of contributing to the understanding of extreme dynamic behaviors, which achieve maximum possible values in parameters of interest such as height, the length of the limit cycle, and the minimum number of Garden of Eden states, among others.

The relationship between the digraph \(G_{\Gamma}\) of Hamiltonian dynamics and the associated interaction graph was demonstrated. In particular, the existence of networks that cannot be modeled using interaction graphs \(G(f)\) with bounded in-degree was proven, requiring specific connectivity conditions to reproduce these dynamics (see \Cref{tab:resumen}).

\begin{table}[htb!]
    \centering
    \begin{tabular}{|c|c|c|c|}
        \hline
        \makecell{\textbf{Type of} \\ \textbf{dynamics}} 
        & \makecell{\textbf{Variable with} \\ \textbf{total dependency}} 
        & \makecell{\textbf{Type of} \\ \textbf{connectivity}} 
        & \makecell{\textbf{Existence of} \\ \textbf{unate} \\ \textbf{network}} 
        \\ \hline
        \makecell{Hamiltonian of \\ maximum height} 
        & Yes 
        & \makecell{Strongly \\ connected} 
        & Yes
        \\ \hline
        \makecell{Hamiltonian \\ intermediate with\\ even period}
        & Yes
        & \makecell{Unilaterally \\ connected}
        & Yes
        \\ \hline
        \makecell{Hamiltonian \\ intermediate with\\ odd period}
        & Yes
        & \makecell{Strongly \\ connected} 
        & Yes
        \\ \hline
        \makecell{Hamiltonian \\ cycle} 
        & Not necessarily
        & \makecell{Unilaterally \\ connected}
        & Yes
        \\ \hline
        \makecell{Quasi- \\ Hamiltonian}
        & \makecell{Not necessarily} 
        & \makecell{Strongly \\ connected} 
        & Unknown
        \\ \hline
        2-Hamiltonian 
        & Not necessarily
        & No restrictions
        & Yes
        \\ \hline
    \end{tabular}
    \caption{Summary of properties present in the dynamics under study.}
    \label{tab:resumen}
\end{table}

Additionally, the inherent limitations of certain families of BNs for modeling Hamiltonian dynamics were analyzed. As a primary contribution, a family of unate networks \(f^{[n]}\) with Hamiltonian dynamics was presented, including cases of maximum height, intermediate height, quasi-Hamiltonian, Hamiltonian cycles, and their generalization to 2-Hamiltonian dynamics. The network \(f^{[n]}\) is notable for being self-dual, suggesting that self-duality in \([n]\) may be a necessary condition for any Hamiltonian cycle network to be unate.

Furthermore, the network \(f^{[n]}\) allows corroboration of the capacity of unate networks to model dynamics with an attractor of arbitrary length without requiring these networks to be bijective. This result broadens the understanding of unate networks and their applications in modeling dynamic systems.

Finally, although the results presented are limited to networks defined over a binary alphabet, the techniques and constructions developed in this work could be generalizable to networks with alphabets of size \(q \geq 2\). This aspect opens the door to new lines of research exploring the extension of these properties to complex systems.

\section*{Acknowledgements}

The authors would like to thank Florian Bridoux and Adrien Richard for their valuable comments and suggestions, which helped to improve this manuscript. This work was partially supported by ANID-Chile through the Center for Mathematical Modeling (CMM), BASAL project FB210005, and by the ANID-Subdirección de Capital Humano, Magíster Nacional 2023, 22231646.

\bibliographystyle{abbrv}
\bibliography{main}
\end{document}